\documentclass[journal,10pt,twocolumn]{IEEEtran}
\IEEEoverridecommandlockouts
\usepackage{cite}
\usepackage{amsmath,amssymb,amsfonts,amsthm}
\usepackage{graphicx}
\usepackage{textcomp}
\usepackage{xcolor}
\usepackage{dsfont}
\usepackage{bm}
\usepackage{acronym}

 \usepackage{multirow}

\usepackage[left=0.625in, right=0.625in, top=0.7in, bottom=1.0in]{geometry}

\usepackage{tabu}
\usepackage{subcaption}
\usepackage{tikz}
\usepackage{lipsum}

\newtheorem{theorem}{Theorem}
\newtheorem{lemma}{Lemma}

\newtheorem{definition}{Definition}

\newtheorem{remark}{Remark}

\usepackage{algorithm}
\usepackage{algpseudocode}

\title{Collaborative Inference over Wireless Channels with Feature Differential Privacy}

\author{
    \IEEEauthorblockN{
        Mohamed Seif \quad Yuqi Nie \quad Andrea J. Goldsmith \quad H. Vincent Poor \\
    }
    \thanks{
      A preliminary version was presented in part at the 57th Asilomar Conference on Signals, Systems, and Computers 2024. This work was supported by the AFOSR award \#002484665, a Huawei Intelligent Spectrum grant, and NSF grants CCF-1908308, CNS-2128448, and ECCS-2335876.
    }
    \IEEEauthorblockA{
        Department of Electrical and Computer Engineering, \\
        Princeton University, Princeton, NJ, 08544 \\
        Email: \{\textit{mseif, ynie, goldsmith, poor}\}@princeton.edu
    }
}

\begin{document}

\maketitle

\begin{abstract}
 Collaborative inference among multiple wireless edge devices  has the potential to significantly enhance Artificial Intelligence (AI) applications, particularly for sensing and computer vision. This approach typically involves a three-stage process: a) data acquisition through sensing, b) feature extraction, and c) feature encoding for transmission. However, transmitting the extracted features poses a significant privacy risk, as sensitive personal data can be exposed during the process. To address this challenge, we propose a novel privacy-preserving collaborative inference mechanism, wherein each edge device in the network secures the privacy of extracted features before transmitting them to a central server for inference. Our approach is designed to achieve two primary objectives: 1) reducing communication overhead and 2) ensuring strict privacy guarantees during feature transmission, while maintaining effective inference performance. Additionally, we introduce an over-the-air pooling scheme specifically designed for classification tasks, which provides formal guarantees on the privacy of transmitted features and establishes a lower bound on classification accuracy.
\end{abstract}


\section{Introduction}

Artificial intelligence (AI) is expected to be a key enabler for new applications in next-generation networks \cite{saad2019vision, letaief2019roadmap, kairouz2021advances}. For example, it can enable low-latency inference and sensing applications, including autonomous driving, personal identification, and activity classification (to name a few). Two conventional AI paradigms are commonly used in practice for these applications: 1) On-device inference that locally performs AI-based inference, which suffers from high computation overhead relative to device capabilities, and 2) On-server inference, where edge devices upload their raw data to a central server to perform a global inference task. The latter approach suffers from high communication overhead. To remedy these challenges, edge-device collaborative inference is a compelling solution. In this setting,  joint inference is divided into three modules: a) sensing for data acquisition, b) feature extraction, and 3) feature encoding for transmission. Leakage of fine-grained information about individuals is a risk that must be considered while designing such collaborative communication systems.

To address this challenge, we develop a new private collaborative inference mechanism wherein each edge device in the network protects the sensitive information of  extracted features before transmission to a central server for inference. The key design objectives of this approach are two-fold: 1) minimizing the communication overhead and 2) maintaining rigorous privacy guarantees for transmission of features over a communication network, while providing satisfactory inference performance. Our wireless distributed machine learning transmission scheme, inspired by the findings in \cite{seif2020wireless}, optimizes bandwidth, computational efficiency, and differential privacy (DP) by leveraging the superposition nature of the wireless channel. This approach, in contrast to tradition orthogonal signaling methods, offers enhanced privacy and expedited task accuracy. Further strengthening our scheme, we incorporate additional novel strategies involving aggregated perturbation coupled with device sampling. This method introduces controlled noise to the aggregated data from multiple devices. The combination of wireless superposition, edge device sampling, and aggregated perturbation forms a comprehensive and efficient transmission framework for the wireless collaborative inference problem (see recent works in the literature \cite{yilmaz2022over, liu2023over, chen2023view}).

\subsection{More Related Works}

A critical research theme in edge AI is collaborative inference, which aims to optimize the execution of inference tasks by leveraging both edge devices and centralized servers. Traditionally, raw data is transmitted from edge devices to a centralized server, where a powerful, complex model performs the inference—this is known as server-based inference. While effective in leveraging high computational resources, this approach suffers from substantial communication overhead, particularly when dealing with large volumes of data, which can lead to latency and increased costs.

To address these challenges, an alternative approach known as device-based inference has emerged, where inference tasks are executed directly on resource-constrained edge devices. This method significantly reduces communication costs and latency by eliminating the need to transmit raw data. However, the trade-off is a potential decline in performance, as the limited computational capacity of edge devices can hinder the execution of complex models, resulting in less accurate or slower inference outcomes.

\textit{Collaborative inference} seeks to balance these two approaches by distributing the inference workload between edge devices and the server. This strategy aims to enhance local processing capabilities while minimizing communication overhead, thereby improving the overall efficiency and effectiveness of edge AI systems.

\subsection{Contributions} 
\label{subsec:contributions}

In this paper, we initiate the study of differentially private collaborative inference over wireless channels. We introduce a novel privacy framework, termed \textit{feature differential privacy}, aimed at safeguarding the extracted features of the target during transmission. Additionally, we provide a theoretical lower bound on classification accuracy as a function of key system parameters, including neural network architectures, channel conditions, transmit power, number of users, and privacy noise levels. Furthermore, we propose two private \textit{feature-aware} transmission schemes that factor in feature quality through an entropy measure, offering improved classification accuracy at the cost of additional privacy leakage when compared to the \textit{feature-agnostic} scheme. To support our theoretical analysis, we conduct numerical experiments that validate the effectiveness of the proposed framework. To the best of our knowledge, this is the first work that rigorously addresses collaborative inference over wireless channels with formal privacy and utility guarantees.

While the existing literature on privacy-preserving distributed learning has primarily focused on protecting data during model training (see comprehensive surveys in \cite{kairouz2021advances} and \cite{ulukus2022private}), there is a gap in understanding the privacy-utility trade-off during inference over wireless channels. This paper addresses this gap by providing theoretical insights and practical validation for privacy protection during inference time.

\subsection{Paper Organization}
\label{subsec:paper_organization}

The remainder of the paper is organized as follows. In Section \ref{sec:system_model}, we introduce the system model and describe the proposed transmission scheme along with its privacy guarantees. Section \ref{sec:transmission_scheme} provides the analysis of mean-squared error (MSE) for the \textit{feature-agnostic} scheme, where device participation is independent of feature quality. In Section \ref{sec:feature_aware_transmission}, we present two variants of the scheme that incorporate feature importance via entropy measures, enhancing classification accuracy. Section \ref{sec:experiments} details numerical simulations to demonstrate the efficacy of the proposed framework. Finally, in Section \ref{sec:conclusions}, we conclude the paper. Detailed derivations and the privacy-constrained weight optimization process are provided in the appendix. \\

\begin{figure*}[t]
    \centering
    \includegraphics[width= 1.5\columnwidth]{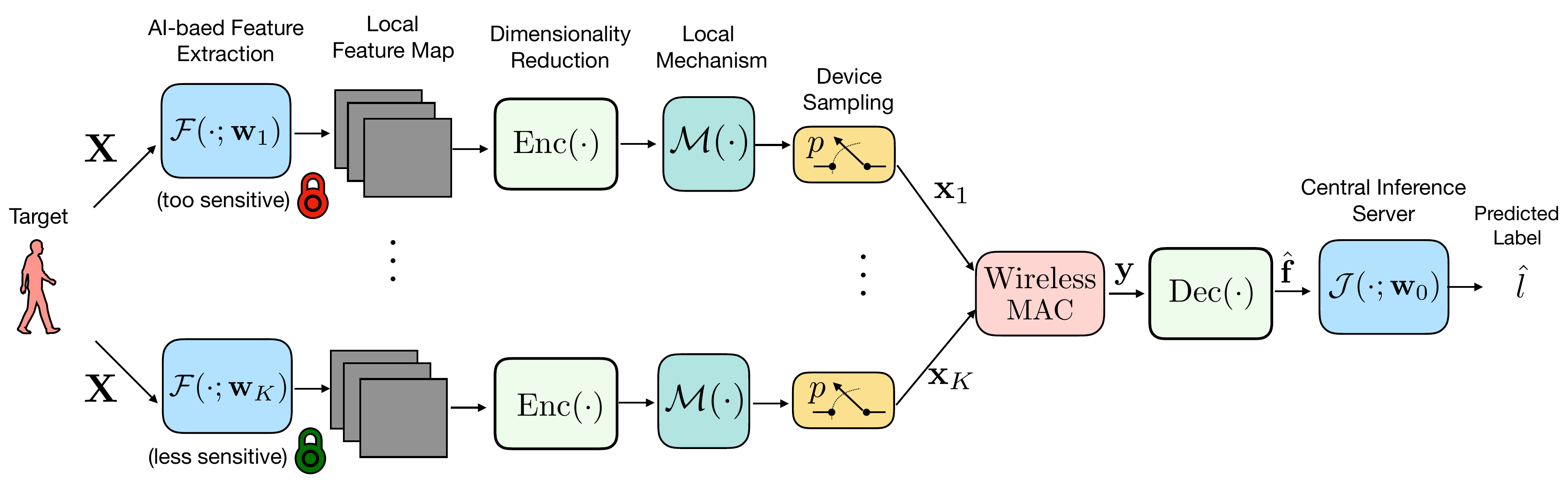}
    \caption{\small{Illustration of the private task-inference framework: Each edge device extracts features from the observed input that preserves some relevant information for classification while satisfies rigorous feature DP levels. Then, each device forwards the processed features over a communication channel to be processed by the central inference server.}}
    \label{fig:proposed_inference_model}
\end{figure*}

\textbf{Notation:} Boldface uppercase letters denote matrices (e.g., $\textbf{A}$),
boldface lowercase letters are used for vectors (e.g., $\textbf{a}$), we denote scalars by non-boldface lowercase letters (e.g., $x$), and sets by capital calligraphic letters (e.g., $\mathcal{X}$). $[K] \triangleq [1, 2, \cdots, K]$ represents the set of all integers from $1$ to $K$.  The set of natural numbers, integer numbers, real numbers and complex numbers are denoted by $\mathds{N}$, $\mathds{Z}$, $\mathds{R}$ and $\mathds{C}$, respectively. $\|\mathbf{x}\|_{2}$ denotes the Euclidean norm of $\mathbf{x}$, and $\|\mathbf{X}\|_{F}$ denotes the Frobenious norm of $\mathbf{X}$.

\section{System Model \& Problem Statement} \label{sec:system_model}

\subsection{Communication Channel Model}

We consider a single-antenna distributed inference system with $K$ edge devices and a central inference server. The edge devices are connected to the inference server through a wireless multiple-access channel with fading on each link. Let $\mathcal{K}$ be a random subset of edge devices that transmit to the server. The input-output relationship can be expressed as
\begin{align}
    \mathbf{y} = \sum_{k \in \mathcal{K}} h_{k} \mathbf{x}_{k} + \mathbf{m},
\end{align}
where $\mathbf{x}_{k} \in \mathds{R}^{d}$ is the transmitted signal by device $k$, $\mathbf{y}$ is the received signal at the edge server, and $h_{k} \geq 0$ is the  channel coefficient between the $k$th device and the server. We assume a block flat-fading channel, where
the channel coefficient remains constant within the duration of a communication block. We denote $\mathbf{m} \in \mathds{R}^{d}$ as the receiver noise whose elements are independent and identically distributed (i.i.d.) according to a Gaussian distribution with zero-mean and variance $\sigma_{m}^{2}$.

\subsection{Distributed Inference Model}

A pre-trained sub-model is deployed on each edge device $k$ that takes the captured image as input and outputs a feature map (or tensor) of real-valued features. Denote the vectorized version of the tensor as $\mathbf{f}_{k} \in \mathds R^{d}$. The edge server performs a multi-view \textit{average} pooling operation on the received local feature maps $\mathbf{f}_{k}$'s to obtain a \textit{global} feature map ${\mathbf{f}}^{*}$ and feeds it to the pre-trained server model to perform a classification task. The average pooled feature is obtained as 
\begin{align}
    \mathbf{f}^{*} & = \frac{1}{K} \sum_{k=1}^{K} \mathbf{f}_{k}.
\end{align}

\section{Main Reuslts \& Discussions} \label{sec:transmission_scheme}

In this section, we first introduce our proposed transmission scheme.  We then outline the scheme's privacy guarantees as described in Theorem \ref{thm:feature_privacy_guarantee}. Finally, we establish a lower bound for the classification accuracy of our approach in Theorem \ref{thm:lower_bound_classification}. We summarize the scheme in Algorithms $1$ and $2$.

\subsection{Proposed Transmission Scheme}

\noindent $(1)$ \textbf{Feature extraction and dimensionality reduction.} Each device $k$ first performs feature extraction\footnote{In this paper, we neglect the data aquisition error.} to obtain an informative representation of the common target  $\mathbf{X}$. This is followed by dimensionality reduction, executed via an encoding operation. The dimensionality reduction process can be represented as
\begin{align}
\mathbf{z}_{k} = \mathbf{W}_{k} \mathbf{f}_{k}(\mathbf{X}),
\end{align}
where $\mathbf{W}_{k} \in \mathds{R}^{r \times d}$ denotes the weight matrix of the encoder for device $k$ where $r \leq d$, and $\mathbf{f}_{k}(\mathbf{X}) \in \mathds{R}^{d}$. \\

\noindent $(2)$ \textbf{Local perturbation noise for privacy.}  Each device $k$ computes a noisy version of its extracted feature as 
\begin{align}
    \tilde{\mathbf{z}}_{k} = w_{k} \mathbf{z}_{k} + \mathbf{n}_{k}, \nonumber 
\end{align}
where $\mathbf{n}_{k} \sim \mathcal{N}(0, \sigma_{k} \mathbf{I}_{q})$ is the artificial noise for privacy. We further assume that the norm of feature vector is bounded by some constant $C_{k} \geq 0$, and in order to ensure that we normalize the feature vector by $C_{k}$, i.e., $\mathbf{z}_{k} := \min \left(1, C_{k}/\| \mathbf{z}_{k} \|_{2}\right)  \cdot \mathbf{z}_{k}$. Finally, $w_{k} \geq 0$ is a weight coefficient of the $k$th device. \\

\noindent $(3)$ \textbf{Pre-processing for transmission.}  The transmitted signal of device $k$ is given as:
\begin{align}
    \mathbf{x}_{k} =  \begin{cases}  \frac{\alpha_{k} }{p_{k}}  \tilde{\mathbf{z}}_{k}, & \text{w.p.}~p_{k}\\
    \mathbf{0}, & \text{otherwise},
    \end{cases}
    \label{eq:inputsignal}
\end{align}
where $\alpha_{k}$ is a scaling factor. If a device is not participating, it does not transmit anything. Note that we multiply the transmitted signal by $1/p_{k}$ to ensure that the estimated signal (i.e., feature map) seen at the server is unbiased. \\

\noindent $(4)$ \textbf{Features aggregation at the edge server.} The received signal at the inference server is given as:
\begin{align}
    \mathbf{y} 
    = \sum_{k \in \mathcal{K}}  \frac{h_{k} \alpha_{k}w_{k}}{p_{k}}  \mathbf{z}_{k}  +  \sum_{k \in \mathcal{K}} \frac{h_{k} \alpha_{k}}{p_{k}} \mathbf{n}_{k} +  \mathbf{m}.\label{eq:output}
\end{align}
All edge devices pick the coefficients $\alpha_{k}$'s to align their transmitted local features. Specifically, each device $k$ picks $\alpha_{k}$ so that ${h_{k} \alpha_{k}}/{p_{k}} = \gamma, \forall k \in \mathcal{K}$, where $\gamma$ represents the chosen alignment constant. Furthermore, the scaling coefficient $\alpha_{k}$ is chosen to satisfy a peak power constraint $P_{k}$, i.e., $\|\mathbf{x}_{k}\|_{2}^{2} \leq P_{k}$, which leads to the following condition:
\begin{align}
    \alpha_{k} = p_{k} \cdot  \min \left\{\frac{\gamma}{h_{k}}, \frac{\sqrt{P_{k}}}{\sqrt{w_{k}^{2} C_{k}^{2} + \|\mathbf{n}_{k}\|_{2}^{2}}} \right\}. 
\end{align}

\noindent $(5)$ \textbf{Post-processing at the edge server.}  Subsequently, the server performs the following sequence of post-processing:
\begin{align}
    \hat{\mathbf{z}} = \frac{1}{\gamma} \mathbf{y} = \sum_{k \in \mathcal{K}} w_{k} \mathbf{z}_{k} + \sum_{k \in \mathcal{K}} \mathbf{n}_{k} + \frac{1}{\gamma}  \mathbf{m}. \label{eqn:post_processing}
\end{align}

\noindent $(6)$ \textbf{Decode the aggregated signal.} The server then decodes the post-processed signal $\hat{\mathbf{z}}$ as follows:
\begin{align}
    \hat{\mathbf{f}} & = \mathbf{D}  \hat{\mathbf{z}} = \mathbf{D} \sum_{k \in \mathcal{K}} w_{k}  \mathbf{W_{k}} \mathbf{f}_{k} + \mathbf{D} \sum_{k \in \mathcal{K}} \mathbf{n}_{k} + \frac{1}{\gamma} \mathbf{D} \mathbf{m},
\end{align}
where $\mathbf{D} \in \mathds{R}^{d \times r}$ is the decoding matrix deployed at the central server.

\begin{algorithm}[t]
\caption{Differentially Private Feature Extraction}
\label{alg:private_extraction}
\begin{algorithmic}[1]
    \State \textbf{Input:} Collect observations $\{\mathbf{X}_{k}\}_{k=1}^{K}$ of the target $\mathbf{X}$ 
    \For{each edge device $k \in \mathcal{K} $ in parallel}
        \State Perform feature extraction on the observed target using the pre-trained model $\mathbf{w}_{k}$: $\mathbf{f}_{k} = \mathcal{F}(\mathbf{X}_{k}; \mathbf{w}_{k})$
        \State Perform dimensionality reduction: $\mathbf{z}_{k} = \mathbf{W}_{k} \mathbf{f}_{k}$
        \State Clip the feature vector: $\mathbf{z}_{k} \gets \min \left(1, \frac{C_{k}}{\| \mathbf{z}_{k} \|_{2}}\right)  \cdot \mathbf{z}_{k}$
        \State Perturb the feature vector via Gaussian mechanism: $\tilde{\mathbf{z}}_{k} \gets w_{k} \mathbf{z}_{k} + \mathbf{n}_{k}$, where $w_{k}$ is a weight coefficient 
    \EndFor
    \State \textbf{Output:} $\tilde{\mathbf{z}}_{1}, \tilde{\mathbf{z}}_{2}, \ldots, \tilde{\mathbf{z}}_{K}$
\end{algorithmic}
\end{algorithm}

\begin{algorithm}[t]
\caption{Features Aggregation and Model Inference}
\label{alg:model_inference}
\begin{algorithmic}[1]
    \State \textbf{Input:}  $\tilde{\mathbf{z}}_{1}, \tilde{\mathbf{z}}_{2}, \ldots, \tilde{\mathbf{z}}_{K}$
        \State The server performs pooling on operation on the received features to obtain a global feature map according to \eqref{eqn:post_processing}
        \State Decode the global feature map: $\hat{\mathbf{f}} = \mathbf{D}  \hat{\mathbf{z}}$
        \State The global feature map is then fed in the ML model at the server for inference: $\hat{l} = \mathcal{J}(\hat{\mathbf{f}}; \mathbf{w}_{0})$
        \State \textbf{Output:} Predicted label $\hat{l}$
\end{algorithmic}
\end{algorithm}

\subsection{Feature differential privacy analysis}

We analyze the privacy level achieved by our proposed scheme that adds artificial noise perturbations to privatize its local data. More precisely, we analyze the privacy leakage under an additive noise mechanism drawn from a Gaussian distribution \cite{dwork2014algorithmic}. We next describe the threat model. \\

\noindent \textbf{Privacy Threat Model}: In the collaborative inference framework, we assume that the central inference server is \textit{honest but curious}. It is honest because it follows the procedure accordingly, but it might learn sensitive information about features. The inference results are released to potentially untrustworthy third parties, heightening privacy concerns. Our focus is on ensuring differential privacy (DP). DP maintains that algorithm outputs (i.e., the task predictions) are indistinguishable when inputs (i.e., the features) differ slightly. Formally, the \textit{feature} DP guarantee can be described as follows: 

\begin{definition} 
[$(\epsilon, \delta)$-feature DP] Let $\mathcal{D} \triangleq \mathcal{F}_{1} \times \mathcal{F}_{2}  \times \cdots \times \mathcal{F}_{K} $ be the collection of all possible features of a common object $\mathbf{X}$.  A randomized mechanism $\mathcal{M}: \mathcal{D} \rightarrow \mathds{R}^{d}$ is $(\epsilon, \delta)$-feature DP if for any two neighboring $D, D' \in \mathcal{F}$, and any measurable subset $\mathcal{S} \subseteq \text{Range}(\mathcal{M})$, we have
\begin{align}
    \operatorname{Pr}(\mathcal{M}(D) \in \mathcal{S}) \leq e^{\epsilon} \operatorname{Pr}(\mathcal{M}(D') \in \mathcal{S}) + \delta.
\end{align}
Here, we say that a pair of  datasets $D, D' \in \mathcal{D}$ are neighboring datasets if $D'$ can be obtained from $D$ by removing one element, i.e., by removing the feature extracted by the $k$th device. The setting when $\delta = 0$ is referred as pure $\epsilon$-feature DP. 
\end{definition}

 We next present the well known Gaussian mechanism for enabling differential privacy in the following definition.
 
\begin{definition}(Gaussian mechanism \cite{dwork2014algorithmic}) \label{defn:Gaussian_mechanism} Suppose a node wants to release a function $f(X)$ of an input $X$ subject to $(\epsilon, \delta)$- feature DP. The Gaussian release mechanism is defined as
\begin{align}
\mathcal{M}(X) \triangleq f(X) + N, \nonumber
\end{align}
where $N \sim \mathcal{N}(0, \sigma^{2} \mathbf{I}_{d})$. If the sensitivity of the function is bounded by $\Delta_f$, i.e., $\| f(x) - f(x')\|_{2}\leq \Delta_f$, any two neighboring inputs $x, x'$, then for any $\delta \in (0,1]$, the Gaussian mechanism satisfies $(\epsilon, \delta)$- feature DP, where 
\begin{align}
    \epsilon = \frac{\Delta_{f}}{\sigma} \sqrt{2 \log \frac{1.25}{\delta}}. \label{eqn:Gaussian_mechanism_perturbation}
\end{align}
\end{definition}

\begin{remark}
Although we have used the classical Gaussian mechanism, it is possible to use other mechanisms such as Gaussian with variance calibration and optimal denoising \cite{balle_improving_Gaussian}.
The challenge lies in deriving closed form expressions as function of the system parameters, which will no longer yield  closed forms.
\end{remark}

\begin{theorem}\label{thm:feature_privacy_guarantee} (Privacy Guarantee) For each edge device $k$ participates with probability $p_{k} \geq 0$ and utilizes local mechanism with a scaling weight $w_{k} \geq 0$.  The privacy guarantee for the $k$th feature is given as
\begin{align}
    \epsilon_{k} \leq \log \left[ 1 + \frac{p_{k}}{1-\delta'} \left( e^{\frac{c_{k}}{\sqrt{\bar{\mu} - t}}} -1 \right)  \right],  \tilde{\delta}_{k} = \delta' + \frac{p_{k} \delta}{1 - \delta'}, \label{eqn:feature_privacy_guarantee}
\end{align}
for any $\delta, \delta' \in (0,1]$ such that $\operatorname{Pr}(|\mu - \bar{\mu}| \geq t) \leq \delta'$  where $\mu \triangleq \sum_{i = 1}^{K} \tau_{i} \sigma_{i}^{2}$, $\tau_{i} \sim \operatorname{Bern}(p_{i})$, $\bar{\mu} \triangleq \sum_{i=1}^{K} p_{i} \sigma_{i}^{2}$, and $c_{k} \triangleq \gamma w_{k} C_{k} \sqrt{2 \log(1.25/\delta)}$. Further, for a given $\delta'$, we choose the parameter $t$ as $    t = \frac{\max_{k} \sigma_{k}^{2}}{{2/\log(2/\delta')}} + \frac{\sqrt{{\max_{k} \sigma_{k}^{2}}/{9}  + 4 (\sum_{k \in [K]} p_{k} (1-p_{k}) \sigma_{k}^{4})/\log(2/\delta')}}{2/\log(2/\delta')}$.
\end{theorem}

\begin{remark} Dimensionality reduction on the feature maps provides a dual benefit: it enhances communication efficiency and reduces sensitivity, thereby amplifying the privacy guarantees.
\end{remark}

 \begin{remark}
Central to our approach is the custom adaptation of privacy guarantees to the feature's varying sensitivity levels. To address the diversity in data sensitivity and privacy needs, we introduce a system of weight coefficients \(w_{k}\) and clipping threshold $C_{k}$ for each feature vector, reflecting their respective DP sensitivities. This enables a tailored privacy protection approach. The development of a device-specific DP leakage metric, \(\epsilon_{k}\), incorporates these customized parameters, allowing for privacy adjustments that align with the distinct sensitivities of the devices' contributed feature vectors.
 \end{remark}

\section{Private Transmission with Self Sampling} \label{sec:feature_agnostic_transmission}

The goal is to analyze the inter-relationship between accuracy and aggregation error due to the randomness of the privacy-preserving perturbation mechanism and sampling procedure. The Mean Squared Error (MSE) can be readily obtained as follows:
\begin{align}
& \operatorname{MSE}   \triangleq \mathds{E} \left[ \|\hat{\mathbf{f}} - \mathbf{f}^{*}\|_{2}^{2} \right] \nonumber \\ 
& \leq d \cdot \|\mathbf{D}\|_{F}^{2} \cdot  \left[ \sum_{k=1}^{K} p_{k} \sigma_{k}^{2} +  \frac{\sigma_{m}^{2}}{\gamma^{2}} \right] \nonumber \\
& +  \sum_{k = 1}^{K} \bigg[ (w_{k}^{2} p_{k} - 2 w_{k} p_{k} + 1) \|\mathbf{D}\|_{F}^{2} \|\mathbf{W}_{k}\|_{F}^{2} \|\mathbf{f}_{k}\|_{2}^{2} \bigg] \nonumber \\
&   +  \sum_{k < j} \bigg[ (p_{k} p_{j} w_{k} w_{j} - p_{k} w_{k} - p_{j} w_{j} + 1) \mathbf{f}_{k}^{T} \mathbf{W}_{k}^{T} \mathbf{D}^{T} \mathbf{D} \mathbf{W}_{j} \mathbf{f}_{j} \bigg], \nonumber 
\end{align}
where the first term in the MSE expression represents the effective noise seen at the inference server, which includes contributions from both channel noise and local perturbation noise introduced for privacy. The second and third terms quantify the approximation error resulting from the application of weight coefficients and the stochastic nature of device participation. It is crucial to highlight that the expectation in the MSE calculation accounts for the randomness introduced by both the variable participation of devices and the variations in local perturbation noise and channel noises. It is worth highlighting that the third term captures the correlations between features $\mathbf{f}_{k}$'s since they are extracted from the same target $\mathbf{X}$.

\begin{remark}
Note that the deployed central server model has an intrinsic classification margin $\Delta$, that is defined as the minimum distance in which the model classifies correctly the pooled feature $\hat{\mathbf{f}}$ when  $ \|\hat{\mathbf{f}} - \mathbf{f}^{*}\|_{2}  \leq \Delta$, formally defined as follows.
\end{remark}

\begin{definition} [Classification Margin \cite{sokolic2017robust}] The classification margin of a target $\mathbf{X}$ represented by $( \mathbf{f}^{*}, {l}^{*})$ measured by a distance $d$ is defined as 
\begin{align}
    \Delta \triangleq \sup \{B: \|{\mathbf{f}} - \mathbf{f}^{*}\|_{2} \leq B \hspace{0.1in} \operatorname{s.t.} \hspace{0.1in} \hat{l}(\mathbf{f}) =  l^{*}, \forall \mathbf{f}  \}.
\end{align}
\end{definition}

We next establish a lower bound for the classification accuracy of our proposed scheme. This approach is grounded in the concept of the \textit{classification margin}, as detailed in \cite{sokolic2017robust}.

\begin{theorem} [Classification Accuracy] \label{thm:lower_bound_classification} The lower bound on the classification accuracy for our proposed privacy-preserving method can be expressed as
\begin{align}
P(\hat{l} = l^{*}) \geq \max \left\{0, P_{0} \cdot \left(1 - \frac{\operatorname{MSE}}{\Delta^{2}}\right)\right\},
\end{align}
where $P_{0}$ represents the classification accuracy in the ideal case (i.e., no communication errors and privacy constraints), and $\Delta$ represents the inherent classification margin.
\end{theorem}

\subsection{Asymptotic Analysis of Classification Accuracy}

In this subsection, we analyze the asymptotic behavior of the lower bound on classification accuracy for our privacy-preserving method, as a function of the number of layers, the number of devices (\(K\)), and the classification margin (\(\Delta\)). The classification accuracy lower bound is given by:
\[
P(\hat{l} = l^{*}) \geq \max \left\{0, P_{0} \cdot \left(1 - \frac{\operatorname{MSE}}{\Delta^{2}}\right)\right\},
\]
where \(P_0\) represents the ideal classification accuracy (without noise or privacy constraints), \(\operatorname{MSE}\) is the mean squared error, and \(\Delta\) is the classification margin.

\subsubsection{Asymptotic Behavior with Respect to Layers}

As the number of layers increases, the MSE grows approximately linearly due to the propagation of noise and increased approximation errors. Specifically, the MSE behaves asymptotically as: $\operatorname{MSE}_L \sim O(L)$.

Substituting this into the classification accuracy bound gives: $P(\hat{l} = l^{*}) \geq P_0 \left(1 - \frac{O(L)}{\Delta^2}\right)$. As \(L \to \infty\), the classification accuracy degrades, eventually approaching zero if the classification margin \(\Delta\) does not grow proportionally with \(L\). Hence, increasing the number of layers decreases the classification accuracy unless \(\Delta\) increases to counterbalance the rise in MSE.

\subsubsection{Asymptotic Behavior with Respect to Devices}

As the number of devices (\(K\)) increases, the MSE decreases asymptotically as $\operatorname{MSE} \sim O\left(\frac{1}{K}\right)$.
This leads to the classification accuracy bound:
\[
P(\hat{l} = l^{*}) \geq P_0 \left(1 - \frac{1}{K \cdot \Delta^2}\right).
\]
As \(K \to \infty\), the classification accuracy approaches the ideal value \(P_0\). Therefore, increasing the number of devices improves the accuracy by reducing the aggregation error, though the reduction becomes sublinear as \(K\) grows large.

To make the above expression more insightful, we next derive a lower bound on the classification margin, formulated as a function of the neural network parameters on the inference server. The lower bound presented in the next lemma follows directly from the analysis of \cite{sokolic2017robust}.

\begin{lemma}
The classification margin of the pre-trained classifier for correctly classifying the target \(\mathbf{X}\) is lower bounded by 
\begin{align}
    \Delta \geq \frac{\sqrt{2} \cdot \min_{l' \neq l^{*}} (\boldsymbol{\delta}_{{l}^{*}} - \boldsymbol{\delta}_{l'})^{T} \cdot \bar{\mathcal{J}}(\mathbf{f}^{*}; \mathbf{w}_{0})}{\prod_{\mathbf{w}_{0}^{(\ell)} \in \mathbf{w}_{0}} \|\mathbf{w}_{0}^{(\ell)}\|_{F}},
\end{align}
where \(\boldsymbol{\delta}_{i} \in \{0,1\}^{|\mathcal{L}|}\) is the Kronecker delta vector with a value of one in the \(i\)-th element,  \(\mathcal{L}\) is the set of all classes, \(\mathbf{w}_{0}^{(\ell)}\) represents the matrix of model parameters at the \(\ell\)-th layer of \(\mathbf{w}_{0}\), and \(\bar{\mathcal{J}}(\mathbf{f}^{*}; \mathbf{w}_{0})\) denotes the vector of classification probabilities.
\end{lemma}
\begin{remark}
This lower bound indicates that the classification margin is determined by the separation between the true class \(l^{*}\) and any other class \(l'\), as reflected in the score \(\sqrt{2} (\boldsymbol{\delta}_{{l}^{*}} - \boldsymbol{\delta}_{l'})^{T} \cdot \bar{\mathcal{J}}(\mathbf{f}^{*}; \mathbf{w}_{0})\). The classes are effectively separated by this imposed score, which is directly influenced by the classifier's output \(\bar{\mathcal{J}}(\mathbf{f}^{*}; \mathbf{w}_{0})\) and the structure of the model parameters \(\mathbf{w}_{0}\). 
\end{remark}

\begin{figure}[t]
	\centering
    {\includegraphics[width=0.75\columnwidth]{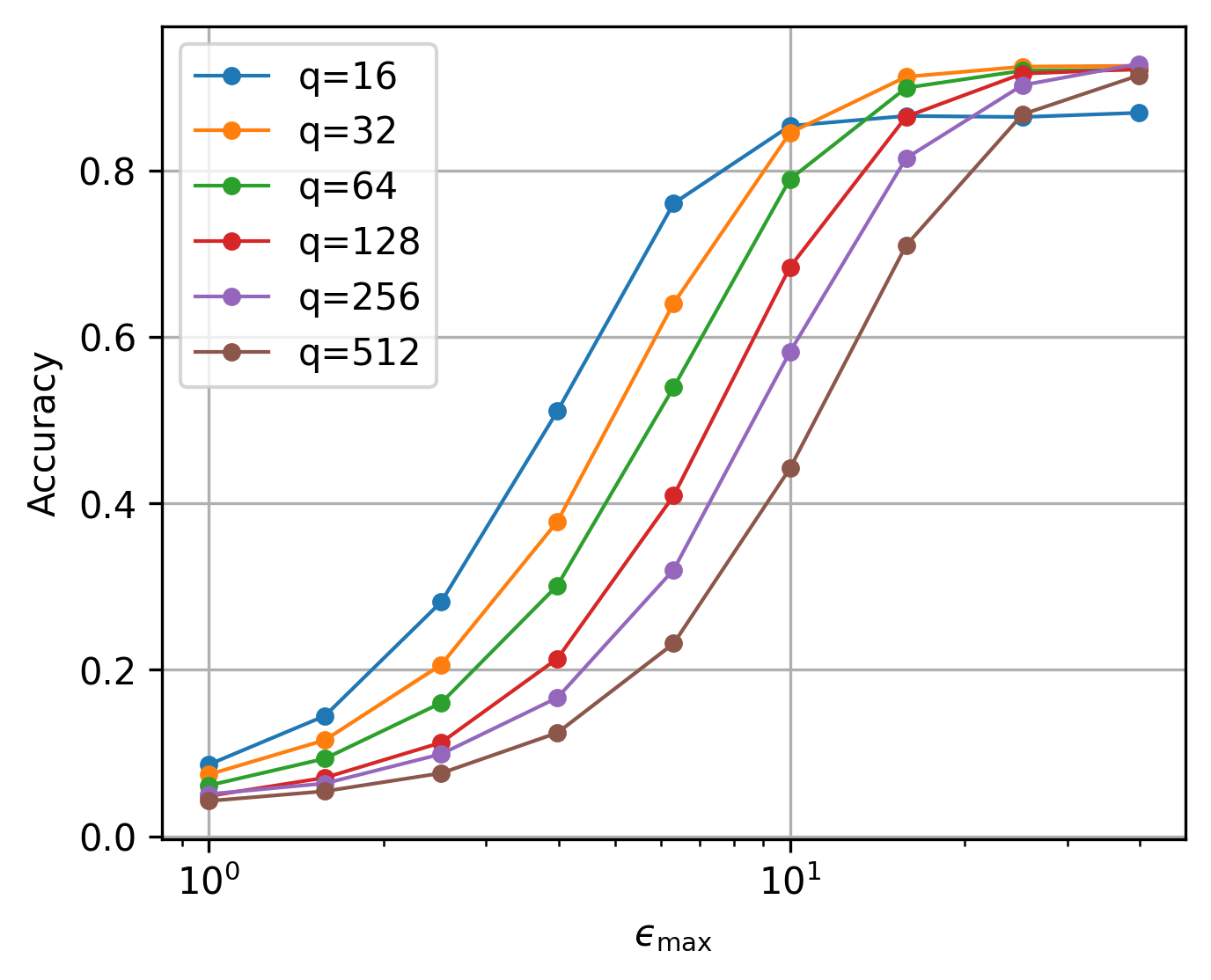}}
    \caption{\small{Impact of the dimensionality reduction on the classification accuracy for the same privacy leakage, where $r = q \times 7 \times 7$.}}
    \label{fig:impact_of_hete_dimensionality}
    \vspace{-10pt}
\end{figure}

\section{Private Feature-aware Transmission} \label{sec:feature_aware_transmission}

In this section, we introduce our private feature-aware transmission scheme, a strategy designed to enhance classification accuracy in the presence of various data acquisition environments, where extracted features are of different importance in terms of \textit{privacy requirements} and \textit{influence of the inference decision}. For example, when an edge device extracts a highly noisy image of an object \(\mathbf{X}\), local inference at the device can help determine whether the device should participate less actively in the task and vice versa. This decision-making process is facilitated through the softmax output values of a lightweight inference module. The softmax output is interpreted as the posterior probability \(p_{\bm{\theta}}(l|\mathbf{f})\), where \(\bm{\theta}\) represents the parameters of the lightweight model used for local inference.

To measure the importance of the feature, we define an entropy function \(u: \mathds{R}^{|\mathcal{L}|} \rightarrow \mathds{R}\), where \(|\mathcal{L}|\) denotes the number of class labels. The edge device transmits its extracted features to the server only if \(u(\mathbf{f}) \leq \eta\), where \(\eta\) is a pre-determined threshold. This selective transmission mechanism ensures that only sufficiently confident features, as determined by local inference, are communicated to the server, optimizing resource use and improving overall system performance. It is worth mentioning that device participation based on the importance of the feature will leak additional information. Hence, we require an extra protection layer, as will be discussed next.

\subsection{Local Initial Inference}
In this scheme, each edge device $k$ locally estimates the inference confidence through the softmax output values of the MLP head. The softmax output can be interpreted as the posterior probability $p_{\bm{\theta}_{k}}(l | \mathbf{f}_{k})$ of the class label $l$ of the target $\mathbf{X}$, where $\bm{\theta}_{k}$ is a shallow local model for initial inference.
We can also quantify the uncertainty of the local device's inference via the Shannon entropy function \cite{shannon1948mathematical}\footnote{We can also quantify the uncertainty of the local device's inference via the min-entropy \cite{konig2009operational}, which is defined as  $u_{\min, k}  = - \log_{2} \max_{l \in \mathcal{L}} p_{\bm{\theta}_{k}} (l | \mathbf{f}_{k})$, where \( u_{\min, k} \) represents the worst-case uncertainty in predicting the label \( l \) given the feature vector \( \mathbf{f}_{k} \). The sensitivity of this function, denoted by \( \Delta u_{\min, k} \), is bounded above by $\Delta u_{\min, k} \leq |- \log_{2} p_{\max} +  \log_{2} p'_{\max} | \leq - \log_{2} p_{\max}$, where \( p_{\max} = \max_{l \in \mathcal{L}} p_{\bm{\theta}_{k}} (l | \mathbf{f}) \) for any \( \mathbf{f} \), and $\mathcal{L}$ is the set of possible class labels. This upper bound indicates that the sensitivity depends on the logarithm of the maximum probability assigned by the model, \( p_{\max} \), which can be further controlled through a clipping step, ensuring that \( p_{\max} \) does not exceed a pre-determined threshold.  
}, which is defined as 
\begin{align}
    u_{k} & = - \sum_{l \in \mathcal{L}} p_{\bm{\theta}_{k}} (l | \mathbf{f}_{k}) \log_{2}p_{\bm{\theta}_{k}} (l | \mathbf{f}_{k}) \label{eqn:shannon_entropy}.
 \end{align}

We further assume that the uncertainty score $u_{k}$ is bounded by some constant $\Gamma_{k} \geq 0$, and in order to ensure that we normalize it by $\Gamma_{k}$, i.e., ${u}_{k} := \min \left(1, \Gamma_{k}/ 2 \log_{2}(|\mathcal{L}|) \right)  \cdot {u}_{k}$. Note that the sensitivity of eqn. \eqref{eqn:shannon_entropy}, $\Delta u_{k} = \max_{\mathbf{f}_{k}, \mathbf{f}_{k}'} |u_{k} - u_{k}'|$ is upper bounded by $2 \log_{2}(|\mathcal{L}|)$. This upper bound may be large for some applications. Therefore, we control the sensitivity by a clipping step.  \\

\begin{figure}[t]
    \centering
    \includegraphics[width= 0.75\columnwidth]{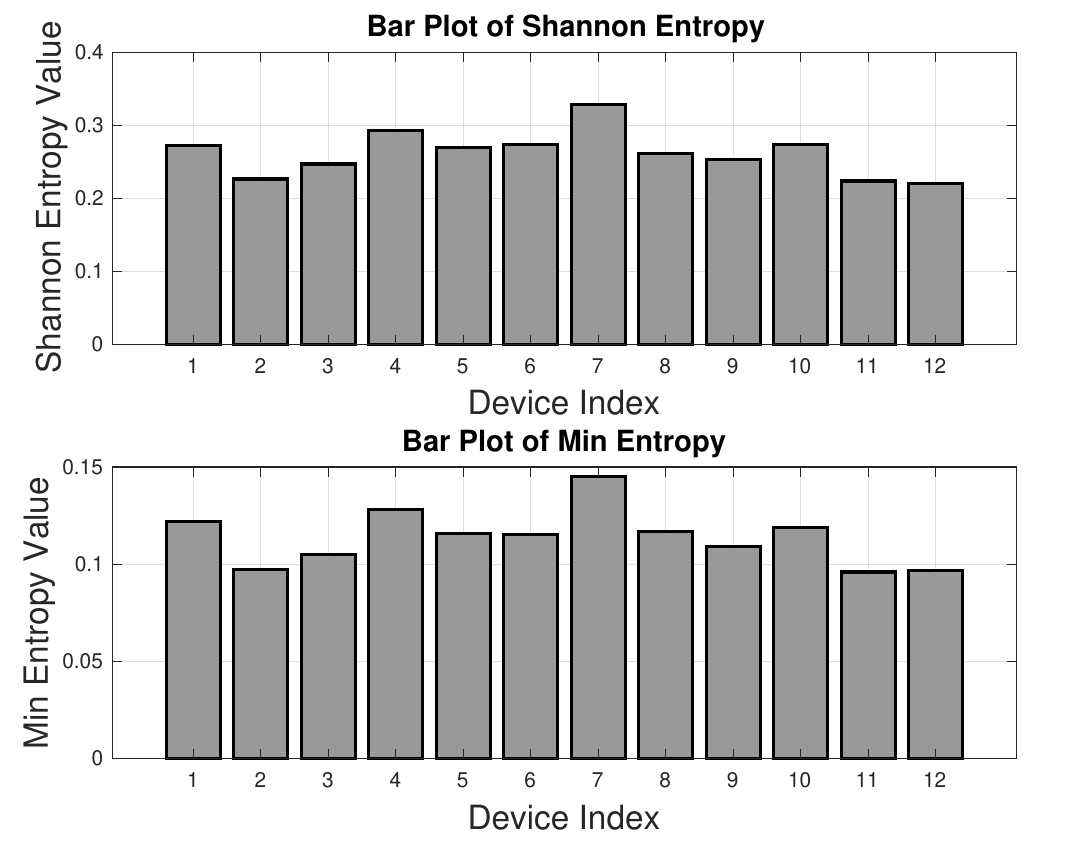}
    \caption{\small{Entropy Comparison Based on the ModelNet Dataset: The top plot illustrates the Shannon entropy values for each device index, while the bottom plot presents the corresponding Min entropy values. Shannon entropy computes the average uncertainty across possible outcomes, whereas Min entropy is a more restrictive measure, focusing on the most probable outcome. Despite their different calculations, both functions exhibit similar behavior across the devices in the ModelNet dataset, suggesting that the choice between the two measures does not significantly alter the device participation strategy.}}
    \label{fig:error_bar_plot}
\end{figure}

\begin{figure}[t]
    \centering
    \includegraphics[width= 0.75\columnwidth]{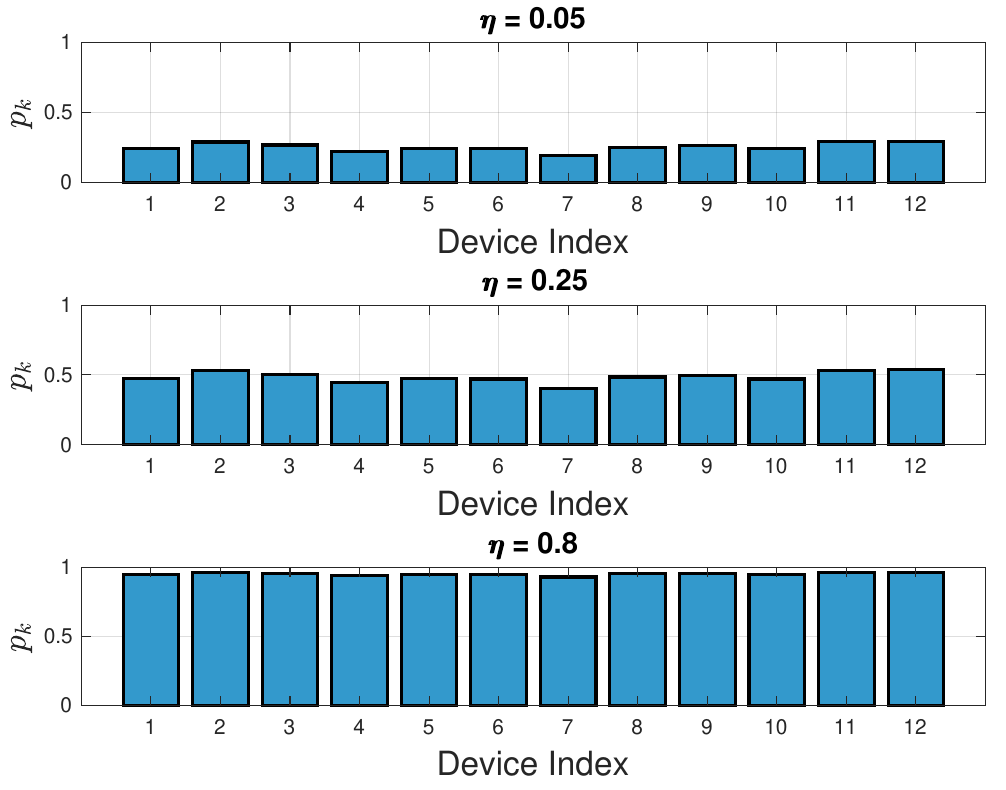}
    \caption{\small{Participation Probability vs. Device Index for Different Thresholds (\(\eta\)) for \textbf{Scheme 2.1}:
        This figure shows the participation probabilities \( p_k \) for each device index based on three thresholds: \(\eta = 0.05\), \(\eta = 0.25\), and \(\eta = 0.8\). 
        The probabilities \( p_k \) are obtained from eqn. \eqref{eq:inputsignal_scheme_2}, with \( u_k \) representing the Shannon entropy for each device and \(\sigma_k^{(0)} = \sqrt{0.1}\). 
        The plots illustrate how varying \(\eta\) affects the likelihood of device participation, with consistent axis limits across subplots for easier comparison.}}
    \label{fig:scheme_2_1_p_k}
\end{figure}

\textbf{Differentially Private Inference Uncertainty Score.} Each edge device $k$ \textit{privately} computes its uncertainty score via the Gaussian mechanism as follows
\begin{align}
    \tilde{u}_{k} & =  u_{k} + v_{k}, \label{eqn:confidence_score} 
\end{align}
where $v_{K}$ is the privacy Gaussian  noise with noise parameter $\sigma^{(0)}_{k}$, which provides $(\epsilon^{(0)}_{k}, \delta^{(0)}_{k})$-feature DP for the $k$th device.

After computing the private local uncertainty scores, we consider two variants for device participation: (1) \textbf{Scheme 2.1} (Local Device Selection), where each device \( k \) locally decides to participate in the inference task if its uncertainty score \( \tilde{u}_{k} \) is below a threshold \( \eta \); and (2) \textbf{Scheme 2.2} (Server-based Device Selection), where devices report their uncertainty scores to a server, which then sorts them and selects the devices with the lowest scores. We shall describe each scheme in detail below. \\

\begin{figure}[t]
    \centering
    \includegraphics[width= \columnwidth]{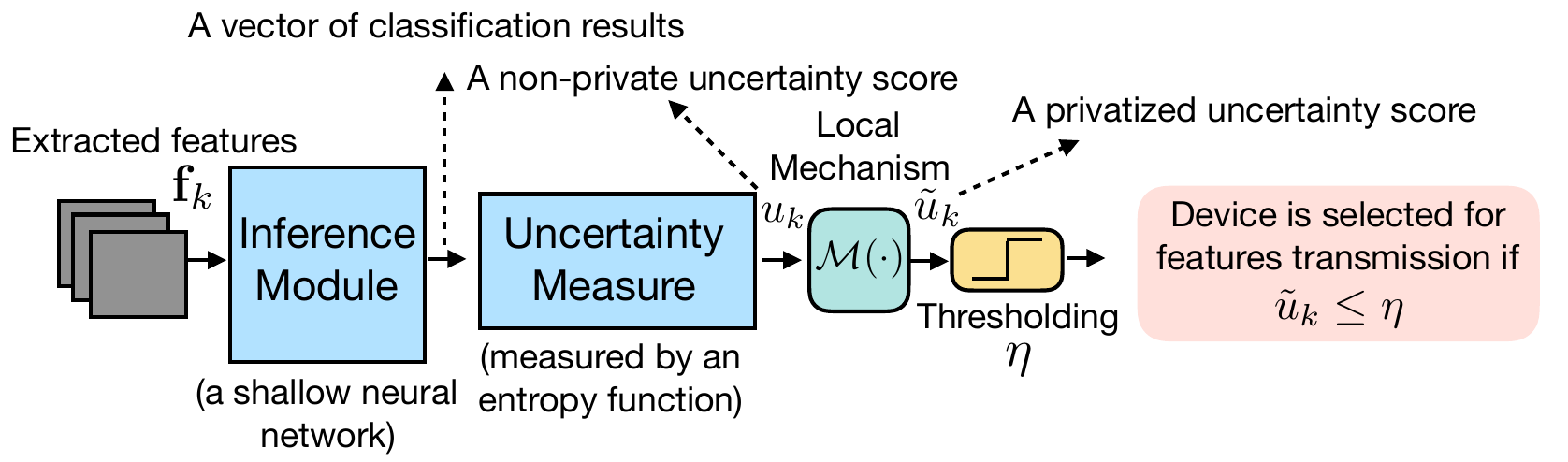}
    \caption{\small{Illustration of the private feature aware transmission for \textbf{Scheme 2.1}. Edge device $k$ participates in the collaborative inference iff the \textit{privatized} uncertainty score $\tilde{u}_{k}$ is below a certain threshold $\eta$.}}
    \label{fig:proposed_feature_aware_inference_model}
\end{figure}

We now describe \textbf{Scheme} $2.1$ in the following subsection.  

\subsection{Scheme $2. 1$: Local Device Selection}

After computing eqn. \eqref{eqn:confidence_score}, device $k$ transmits its privatized and compressed feature $\tilde{z}_{k}$ as follows:
\begin{align}
    \mathbf{x}_{k} =  \begin{cases}  \frac{\alpha_{k} }{{p}_{k}}  \tilde{\mathbf{z}}_{k}, & \text{if}~ \tilde{u}_{k} \leq \eta\\
    \mathbf{0}, & \text{otherwise},
    \end{cases}
    \label{eq:inputsignal_scheme_2}
\end{align}
where $\eta$ is a pre-specified threshold, and the device sampling probability $p_{k}$ can be readily shown to be ${p}_{k} = \Phi \left((\eta - u_{k})/ \sigma_{k}^{(0)}\right)$. The selected set of devices $\tilde{K}$  whose noisy uncertainty scores are below a threshold $\eta$,  is defined as 
\begin{align}
    \tilde{\mathcal{K}} &= \{i:i \in [K] \hspace{0.05in} \text{s.t.} \hspace{0.05in} u_{i} \leq \zeta \}. 
\end{align}

\begin{remark}
 One also can use the Laplacian mechanism to obtain pure DP, however, we prefer to use the Gaussian mechanism to obtain more insightful analysis as a function of the heterogeneous amount of privacy noises.
\end{remark}

 We next propose a centralized approach, where the server collects all uncertainty scores, evaluates them collectively, and selects the optimal set of devices to participate in the inference task.

\subsection{Scheme $2. 2$: Server-based Device Selection} 

The sorting server selects the edge devices with the lowest \textit{privatized} uncertainty scores $u_{k}$'s for the features transmission to the inference server. Let $\{\tilde{u}_{(1)}, \tilde{u}_{(2)}, \cdots, \tilde{u}_{(K)} \}$ be the sorted uncertainty scores of the $K$ devices.  The server selects only $|\tilde{\mathcal{K}}| =  \tilde{k}$ out of $K$ devices. The selected set of devices of lowest $\tilde{k}$ scores is denoted as 
\begin{align}
    \tilde{\mathcal{K}} & = \{ i: \tilde{u}_{i} \in \{\tilde{u}_{(1)}, \tilde{u}_{(2)}, \cdots, \tilde{u}_{(\tilde{k})} \}\}.
\end{align}
Furthermore, device $k$ is selected if  $k \in \tilde{\mathcal{K}}$.

\begin{table}[t]
    \centering
\begin{tabular}{|c|c|c|c|}
    \hline
    & \textbf{Scheme 2.1 (\(\eta = 0.5\))} & \textbf{Scheme 2.1 (\(\eta = 0.3\))} & \textbf{Scheme 2.2} \\
    \hline
    \rule{0pt}{3ex} 
    $\epsilon_{k}^{(0)}$ & 0.2641 ± 0.0333 & 0.1134 ± 0.0241 & 0.9 ± 0 \\
    \hline
\end{tabular}
    \caption{\small{Summarized Privacy Leakage \(\epsilon_k^{(0)} \) across devices for different schemes and \(\eta\) values on the {ModelNet} dataset. 
The table presents the average leakage, variation (deviation), and range across devices for each scheme, where \(\sigma_{k}^{(0)} = \sqrt{0.1}, \forall k\). 
Note that for Scheme 2.2, the upper and lower bounds coincide for any \(\sigma_{k}^{(0)} \geq 0.3\). }}
    \label{tab:privacy_leakage_summary}
\end{table}

\begin{remark} The private sorting mechanism requires less communication than the actual feature transmission. Communication between edge devices and the sorting server can occur over an encrypted side channel, although we still privatize the communication channel between edge devices and the server.
\end{remark}

\begin{algorithm}
\caption{Differentially Private Uncertainty Score Ranking}
\label{alg:private_sorting}
\begin{algorithmic}[1]
    \State \textbf{Input:} Collect the uncertainty scores $\{u_{k}\}_{k=1}^{K}$
    \For{each edge device $k \in [K] $ in parallel}
        \State Perform a local inference via  local inference model $\bm{\theta}_{k}$
        \State Privately compute the uncertainty score $u_{k}$ in eqn. \eqref{eqn:confidence_score}
        \State Send $\tilde{u}_{k}$ to the sorting server  
        \State Server selects edge devices with the lowest uncertainty scores
    \EndFor
    \State \textbf{Output:} The set of transmitting edge devices $\tilde{\mathcal{K}}$ 
\end{algorithmic}
\end{algorithm}

\begin{theorem} \label{thm:sorting_correctness} The $k$th device participation probability $p_{k}$ is lower bounded as
\begin{align}
 p_{k} \geq \max \left\{ 0, 1 -  \sum_{k=1}^{K-1}  \Phi \left(\frac{-\Psi}{\sqrt{(\sigma^{(0)}_{k})^{2} + (\sigma^{(0)}_{k+1})^{2}}} \right) \right\},
\end{align}
where  $\Phi(\cdot)$ is the cumulative CDF of Gaussian distribution, and $\Psi \triangleq \min_{1 \leq k \leq K-1} u_{(k+1)} - u_{(k)}$. Also, the participation probability $p_{k}$ is upper bounded by
\begin{align}
p_{k} & \leq \prod_{k=1}^{K-1} \Phi\left(\frac{\Psi}{\sqrt{(\sigma^{(0)}_{(k+1)})^2 + (\sigma^{(0)}_{(k)})^2}}\right) \cdot \mathbb{I}(k \in \mathcal{K}) \nonumber \\
& \hspace{0.3in} + \Phi\left(\frac{u_k - \min_{j \in \mathcal{K}^c} u_j}{\sqrt{(\sigma_{k}^{(0)})^2 + (\sigma_{j}^{(0)})^2}}\right),
\end{align}
where $\sigma_{(1)}^{(0)} \leq  \sigma_{(2)}^{(0)} \leq \cdots \leq \sigma_{(K)}^{(0)}$.
\end{theorem}

\begin{figure}[t]
    \centering
    \includegraphics[width= 0.75\columnwidth]{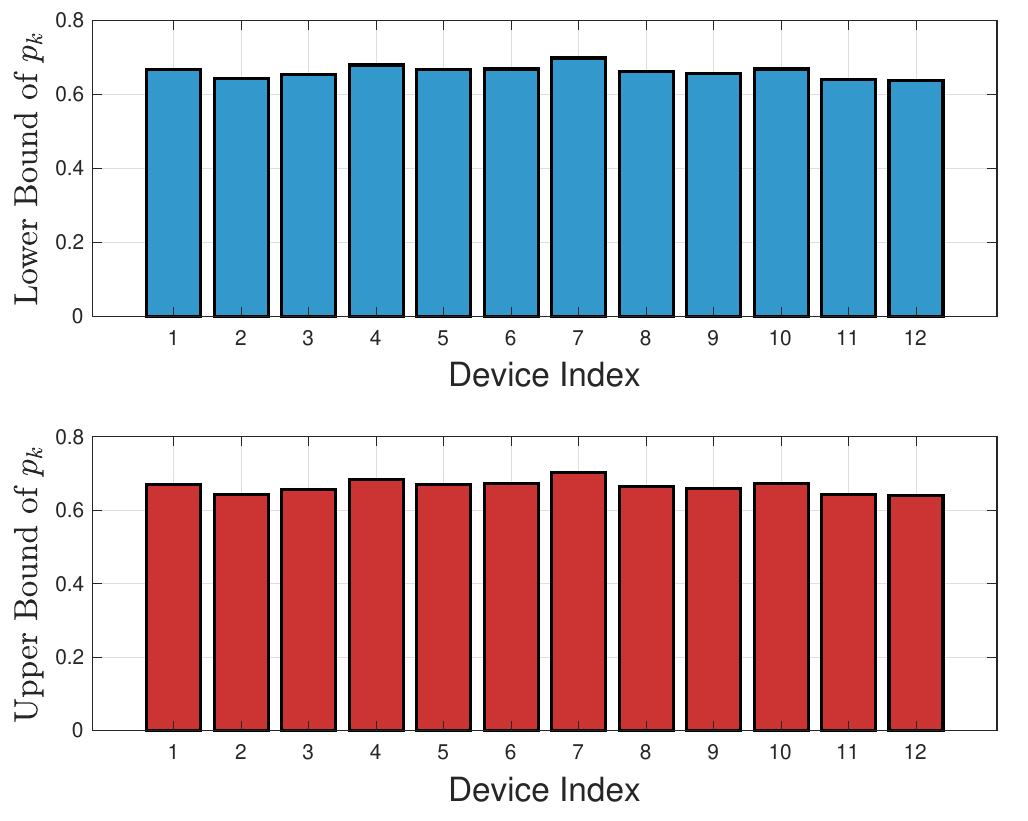}
    \caption{\small{Upper and Lower Bounds of Participation Probability \( p_k \) vs. Device Index with Noise Proportional to Device Reliability.  
    This figure displays the computed upper and lower bounds of the participation probability \( p_k \) for each device index using \(\Psi = 0.005\) and \(\sigma_k\) values that are inversely proportional to the reliability of each device. 
    The noise values \(\sigma_{k}^{(0)} \) are given by \(\sigma_{k}^{(0)} = \frac{1}{u_{k} + 0.01}\), where the mean reliability values are based on device performance. 
    Devices with higher reliability (larger mean values) are assigned lower noise, while devices with lower reliability are assigned higher noise. 
    The upper bounds (bottom plot) are derived from the product of Gaussian CDFs, while the lower bounds (top plot) are calculated using the sum of CDF terms.} }
    \label{fig:scheme_2_2_bar_plot}
\end{figure}

We next state the privacy guarantee for the proposed schemes in the following lemma.

\begin{lemma}\label{thm:feature_privacy_guarantee_new} (Privacy Guarantee) Under both schemes, the privacy guarantee for the $k$th device is $(\epsilon^{(0)}_{k} + \epsilon_{k}, \delta_{k}^{(0)} + \tilde{\delta}_{k})$, where  $(\epsilon_{k}, \tilde{\delta}_{k})$ are obtained from eqn. \eqref{eqn:feature_privacy_guarantee} with different participation probabilities $p_{k}$.
\end{lemma}
Here, we leverage the sequential composition property of DP \cite{dwork2014algorithmic} to establish privacy guarantees. The additional privacy leakage arises from the dependence of the participation probability on the features.

\begin{remark} [Impact of Privacy Noise on Sorting] When the noise terms \( n_k \) are large relative to the differences between the true scores \( u_k \), the selection of devices becomes less sensitive to the values of \( u_k \). As a result, the probability that a particular device \( k \) is selected approaches the uniform probability \(\frac{|\mathcal{K}|}{K}\). In this scenario, the noise effectively obscures the differences in \( u_k \), leading to a near-random selection process.  The term \(\frac{|\mathcal{K}|}{K}\) represents the proportion of devices that are selected from a total of \( K \) devices. It serves as a baseline probability for selection when the selection process is uniform and independent of the devices' true scores \( u_k \). In contrast, when the noise terms \( n_k \) are small, the selection process is predominantly influenced by the true ranking of the scores \( u_k \). Devices with lower values of \( u_k \) are more likely to be included in \(\mathcal{K}\), and the probability of selection for such devices exceeds \(\frac{|\mathcal{K}|}{K}\). In this case, the term \(\frac{u_k - \min_{j \in \mathcal{K}^c} u_j}{\sqrt{(\sigma_{k}^{(0)})^2 + (\sigma_{j}^{(0)})^2}}\) is larger, resulting in a higher value of \(\Phi(\cdot)\).
\end{remark}

\begin{figure}[t]
    \centering
    \includegraphics[width= 0.75\columnwidth]{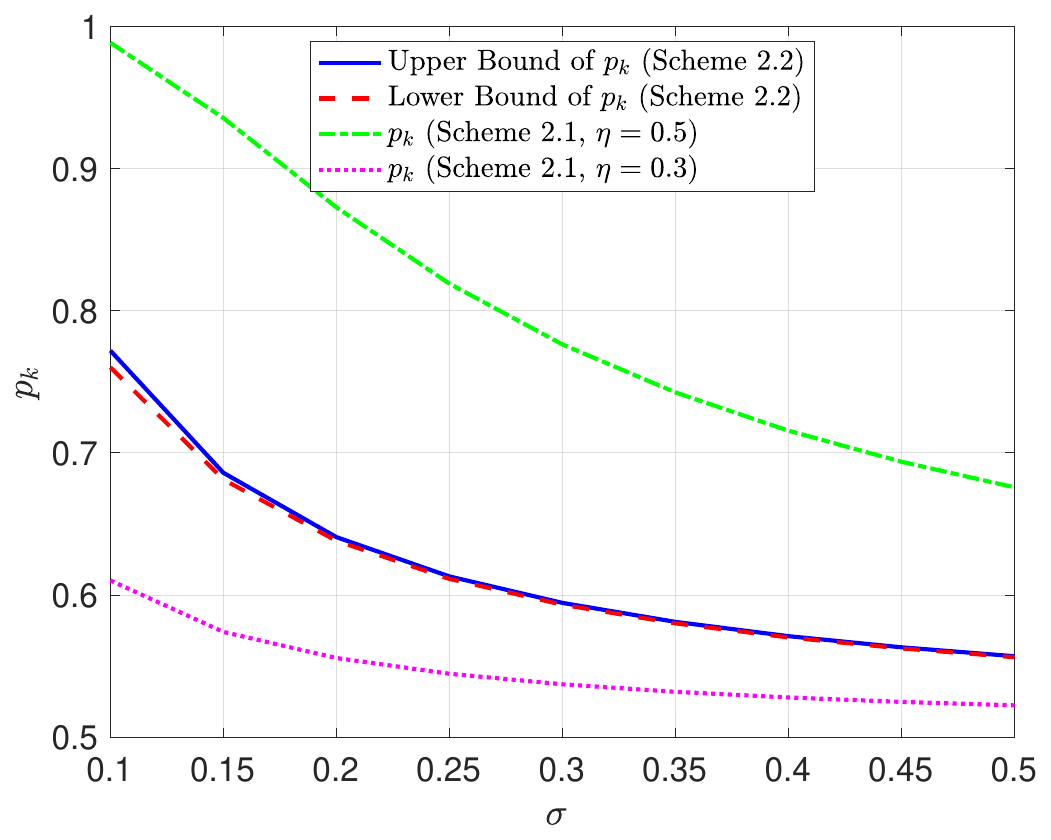}
    \caption{\small{Upper and Lower Bounds of \( p_k \) as Functions of the Noise Standard Deviation \(\sigma\) for \(\Psi = 0.1\) and \(K = 12\).} 
        The upper bound (solid blue line) represents the probability of selecting device \( k \) based on the likelihood of correctly identifying the true set \(\mathcal{K}\) (Scheme 2.2). 
        The lower bound (dashed red line) accounts for scenarios where the true set is not correctly identified (Scheme 2.2). 
        The green dash-dot line represents \( p_k \) calculated using Scheme 2.1 with \(\eta = 0.5\), while the magenta dotted line corresponds to \( p_k \) from Scheme 2.1 with \(\eta = 0.3\). 
        These results illustrate how different schemes and parameter settings affect the selection probabilities across a range of noise standard deviations.}
    \label{fig:error_bar_plot}
\end{figure}

\begin{figure}[t]
	\centering
    {\includegraphics[width=0.75\columnwidth]{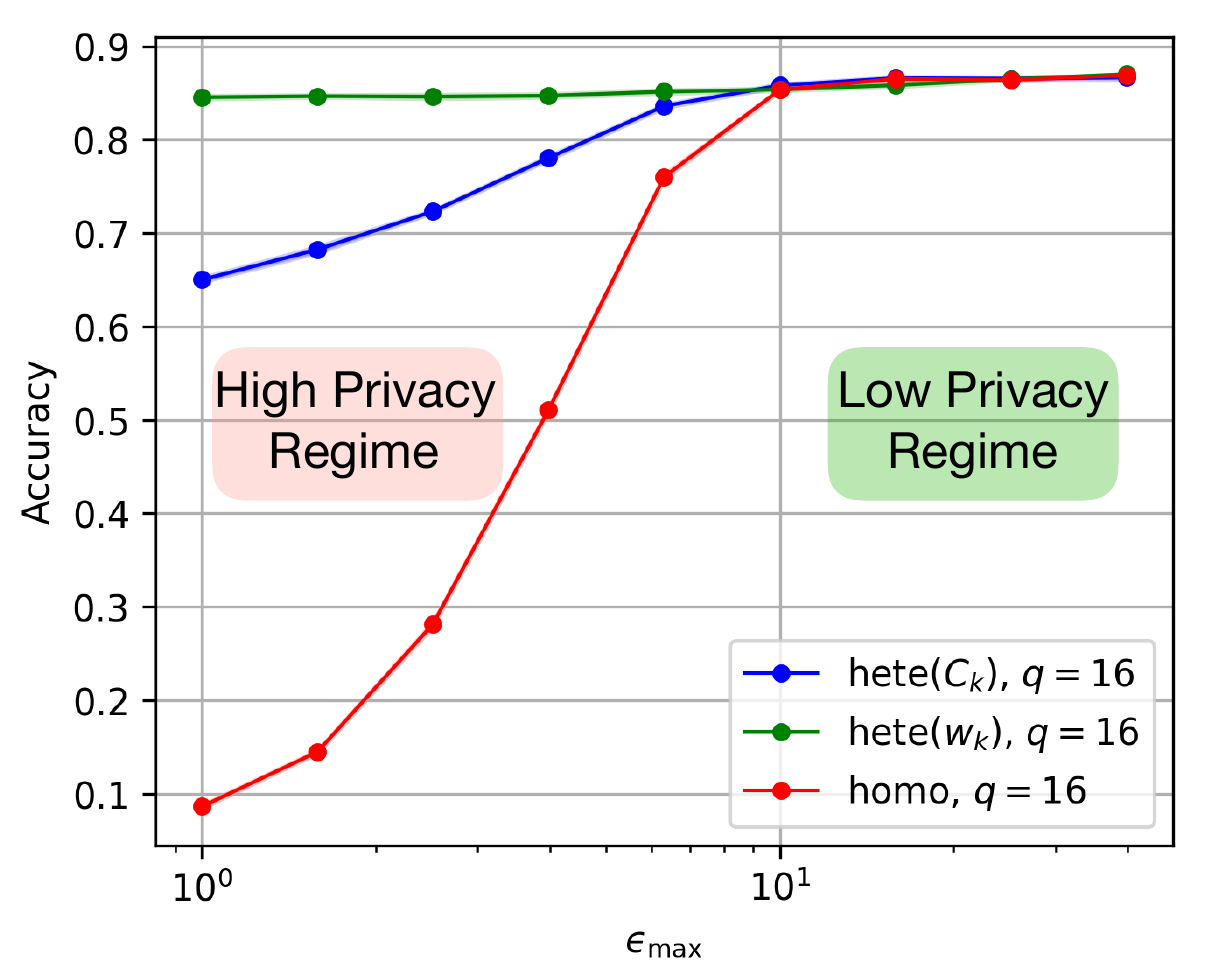}}
    \caption{\small{Impact of customizing privacy levels on the classification accuracy for $r = q \times 7 \times 7$, where $q = 16$.}}
    \label{fig:impact_of_hete_privacy_small}
    \vspace{-10pt}
\end{figure}

\begin{figure}[t]
	\centering
    {\includegraphics[width=0.75\columnwidth]{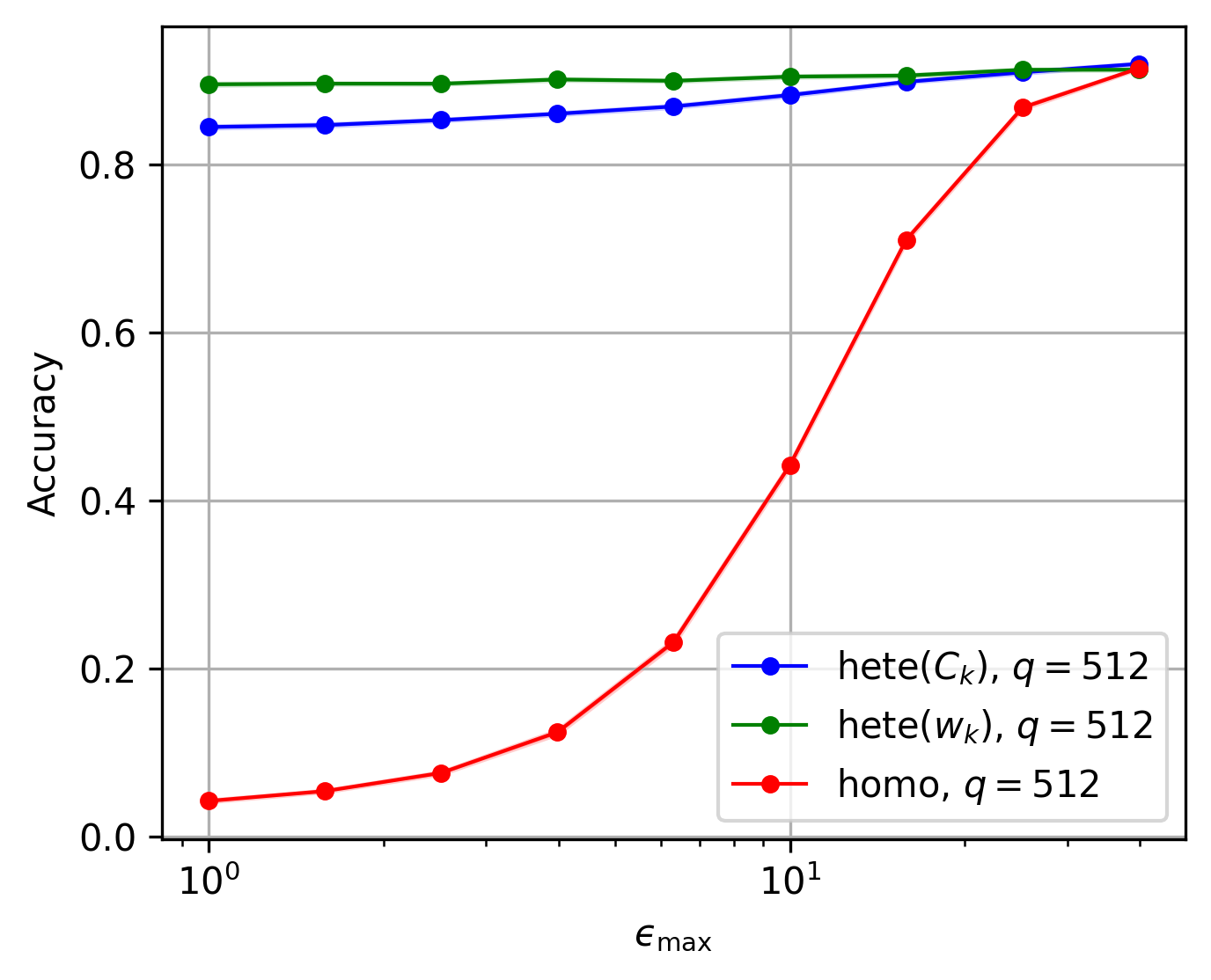}}
    \caption{\small{Impact of customizing privacy levels on the classification accuracy for $r = q \times 7 \times 7$, where $q = 512$.}}
    \label{fig:impact_of_hete_privacy_large}
    \vspace{-10pt}
\end{figure}

\begin{figure}[t]
	\centering
    {\includegraphics[width= 0.8 \columnwidth]{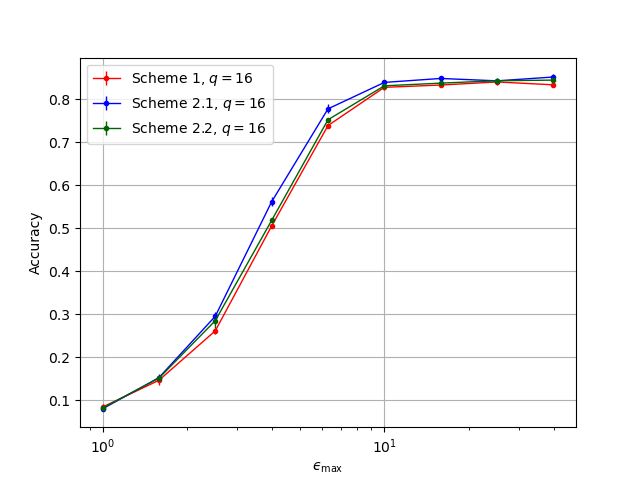}}
    \caption{\small{Comparison between feature-agnostic and feature-aware transmission scheme, where $\sigma_{k}^{(0)} = \sqrt{0.01}$,  $r = q \times 7 \times 7$, and $q = 16$.}}
    \label{fig:comparisons_between_schemes_vs_epsilon_q_16}
    \vspace{-10pt}
\end{figure}

\section{Numerical Experiments} \label{sec:experiments}




In this section, we conduct experiments to assess the performance of the proposed private collaborative inference scheme. We adopt a Rician channel model with a variance of $\sigma^{2}_{m} = 0.1$ to simulate fading channels \cite{goldsmith2005wireless}. Our setup includes $K = 12$ devices, each with a default transmit power of $P_{k} = 30$ dBm, a connection probability of $p_k = 0.9$, and an equal weight of $w_{k} = 1/K = 1/12$. The perturbation noise level is set at $\sigma^{2}_{k}=0.1$, with privacy parameters $\delta = 10^{-5}$ and $\delta'= 10^{-5}$, and an alignment constant of $\gamma=1$. Additionally, we employ feature clipping on each device before transmitting the local feature to a central server, with a clipping threshold of $C_k = 10^2$.

The encoding and decoding matrices are implemented via one hidden layer neural networks for two different dimensions, $16 \times 7 \times 7$ and $512 \times 7 \times 7$, to facilitate the proposed private collaborative inference scheme\footnote{It is worth highlighting that performance worsens for higher dimensions due to the increase in perturbation noise.}. We develop a Multi-View Convolutional Neural Network (MVCNN) architecture utilizing the ModelNet dataset, known for its multi-view images of objects such as sofas and tables, and integrating the VGG11 model. In our design, the VGG11 model is partitioned prior to the linear classifier stage, positioning the classifier at the server for ultimate decision-making and distributing the remaining VGG11 components across sensor nodes for feature extraction, optimized for average pooling. Our study concentrates on a ModelNet subset encompassing $40$ object classes, captured using an array of $12$ cameras arranged to provide a $30$-degree separation between adjacent sensors for thorough and diverse object perspectives. Each sensor, equipped with the adapted VGG11 model, produces feature maps as tensors of dominions $16 \times 7 \times 7$ and $512 \times 7 \times 7$, respectively.

In Fig. \ref{fig:impact_of_hete_privacy_small} and \ref{fig:impact_of_hete_privacy_large}, we reveal the critical need for customizing privacy levels to optimize performance, acknowledging that not all transmitted features bear the same sensitivity. With half of the edge devices processing data with sensitive attributes and the other half not, we highlight two strategic avenues for enhancement: \textit{optimizing} the weight coefficients \(w_{k}\)'s, or \textit{refining} the clipping parameters \(C_{k}\). Upon comparing these approaches with uniform privacy models, where \(w_{k} = 1/K\) and \(C_{k} = C\), our methods demonstrate superior performance in scenarios demanding stringent privacy. It is noteworthy that uniform privacy approaches only approximate the efficacy of our tailored strategies in the low-privacy regime (i.e., large $\epsilon$). This insight aligns with findings from \cite{shlezinger2022collaborative}, which critique the limitations of incorporating DP in inference systems where data privacy prevails, thus highlighting that traditional DP mechanisms may not always ensure optimal utility.

In Fig. \ref{fig:comparisons_between_schemes_vs_epsilon_q_16}, we compare the performance of the feature-agnostic scheme with the two feature-aware schemes, denoted as Schemes 2.1 and 2.2. We set the parameters as follows: $\eta = 5$, $\sigma_k^{(0)} = \sqrt{0.01}$, $\delta_{k}^{(0)} = 10^{-5}$, and $\tilde{k} = 11$. As illustrated in the figure, the feature-aware schemes demonstrate improved classification accuracy at the cost of increased privacy leakage, as determined in Lemma \ref{thm:feature_privacy_guarantee_new}, which is linked to the sampling probability—a function of the sensitive features. Additionally, we observe that under certain design parameters, Scheme 2.1 outperforms the server-based Scheme 2.2, as it requires less coordination between the devices and the server. The detailed accuracy results are summarized in Table \ref{tab:test_accuracy_schemes}.


\begin{table}[t]
\centering
\begin{tabular}{|c|c|c|c|}
\hline
$\epsilon_{\max}$ & Scheme 1 & Scheme 2.1  & Scheme 2.2  \\
\hline
3.9811  & $0.5051 \pm 0.0044$  & $\mathbf{0.5614 \pm 0.0102}$  & $0.5181 \pm 0.0043$ \\
6.3096  & $0.7386 \pm 0.0050$  & $\mathbf{0.7774 \pm 0.0107}$  & $0.7521 \pm 0.0039$ \\
10.0000 & $0.8275 + 0.0070$  & $\mathbf{0.8389 \pm 0.0056}$  & $0.8310 \pm 0.0084$ \\
\hline
\end{tabular}
\caption{\small{Comparison of Accuracy Across Three Schemes. Each scheme reports the mean accuracy along with the standard deviation, where $\sigma_{k}^{(0)} = \sqrt{0.01}$.}}
\label{tab:test_accuracy_schemes}
\end{table}

\section{Conclusions \& Future Work} \label{sec:conclusions}


In this paper, we initiated the study of differentially private collaborative inference over wireless channels, introducing a novel framework called \textit{feature differential privacy} to protect extracted features during transmission. We derived a theoretical lower bound on classification accuracy based on key system parameters and proposed two private \textit{feature-aware} transmission schemes, which improve classification accuracy by accounting for feature quality, though with additional privacy leakage compared to a \textit{feature-agnostic} approach. Our theoretical insights were validated through numerical experiments, confirming the effectiveness of the proposed framework. This work is among the first to rigorously address collaborative inference over wireless channels, providing formal privacy and utility guarantees during inference. While previous research focused on privacy during model training, our study fills the gap by addressing the privacy-utility trade-off during inference, offering both theoretical and practical advancements in privacy-preserving collaborative inference.

\bibliographystyle{IEEEtran}
\bibliography{myreferences}

\appendices
\renewcommand{\thesectiondis}[2]{\Alph{section}:}

\section{Proof of Theorem \ref{thm:feature_privacy_guarantee}}

\subsection{Bernstein's inequality}
\label{subapp:bernstein_inequality}

\begin{lemma} (Bernstein's Inequality \cite{mitzenmacher2017probability}) 
\label{lem:bernstein_inequality}
Let $\{X_k\}_{k=1}^{K}$ be a collection of zero-mean independent  random variables, where each $X_k$ is bounded by $M$, almost surely. Then, for any $t \geq 0$, we have 
\begin{align}
    \operatorname{Pr} \bigg(\sum_{k=1}^{K} X_k > t \bigg) \leq \exp \left[-\frac{r^{2}}{\sum_{k=1}^{K} E(X_k^2) + \frac{Mt}{3}}\right].
\end{align}
\end{lemma}

{\bf Choice of the parameter $t$ for the result in Theorem \ref{thm:feature_privacy_guarantee}}: 
Consider the random variables to be ${X_k = (\tau_{k} - p_{k}) \sigma_{k}^2}$, where  $\tau_{k} \sim \operatorname{Bern}(p_{k})$.
Then, $\mathds{E}[X_k] = 0$, $M = \max_{k \in [K]}\sigma_{k}^2$, and ${\mathds{E}[X_k^2] = p_{k}(1 - p_{k})\sigma_{k}^4}$.
An application of Bernstein inequality gives
%
\begin{align*}
    &\Pr(|\mu - \bar{\mu}| \geq t) = \Pr \left(\left\lvert\sum_{k \in [n]}(\tau_{k} - p_{k}) \sigma_{k} ^2\right\rvert \geq t \right) \nonumber\\
    &\leq 2\exp\left[ -\frac{t^2}{\sum_{k =1 }^{K} p_{k}(1-p_{k})\sigma_{k}^4+ \frac{t \max_{k \in [K]} \sigma_{k}^2}{3}} \right].
\end{align*}
%
For any $\delta' \in (0,1]$, this is a quadratic in $t$ and the right hand side expression can be tightened by choosing $t$ to be

\begin{align}
    t &= \frac{\log(2/\delta')}{2}\left({\max_{k \in [K] } \sigma_{k}^{2}}/{3} \right. \nonumber\\
    &\left. + \sqrt{{\max_{k \in [K] } \sigma_{k}^{2}}/{9}  + \frac{4}{\log(2/\delta')} \left(\sum_{k = 1}^{K}  p_{k} (1-p_{k}) \sigma_{k}^{4}\right)}\right).
\end{align}

 We next prove the privacy amplification due to self sampling of the edge devices.  Let $\mathbf{y}$ denote the output seen at the PS through the Gaussian MAC and $\mathbf{y}_{-k}$ denote the output when device $k$ does not participate. Recall that DP guarantees that any post-processing done on the received signal does not leak more information about the input. Therefore, it is sufficient to show the following,
\begin{align}
    \operatorname{Pr}(\mathbf{y} \in \mathcal{S}) \leq e^{\epsilon_{k}} \operatorname{Pr}(\mathbf{y}_{-k} \in \mathcal{S}) + \tilde{\delta}_k,\label{eq:genericEpsilonC}
\end{align}
and obtain $\epsilon_k$. The main challenge of this proof is the random participation of edge devices and that the local privacy noises get aggregated over the wireless channel. The number of participants $|\mathcal{K}|$ determines the noise amplification of the feature DP guarantee through wireless aggregation.

 To take all possible $\mathcal{K}$ into account for the analysis, we condition the lefthand side of Eqn. \eqref{eq:genericEpsilonC} with the event that $\mathcal{K}$ deviates from the mean, that is, $|\mu - \bar{\mu} | \geq t$ for any $t > 0$, and bound it using Bernstein's inequality and the DP guarantee via the local Gaussian mechanism. Furthermore, we need additional conditioning for the event $\mathcal{E}_{k}$ that denotes the event in which the device $k$ participates in the inference, i.e. $k \in \mathcal{K}$. Note that $p_k=\operatorname{Pr}(\mathcal{E}_{k})$ and the conditional probabilities $ {\bar{p}_k=\operatorname{Pr}(\mathcal{E}_{k}||\mu - \bar{\mu} | < t)}$ can be readily bounded by $p_k$'s using the total probability theorem and Bernstein's inequality, i.e., $\bar{p}_{k} \leq p_{k}/(1-\delta')$. Now, we have
\begin{align}
    \operatorname{Pr}( \mathbf{y} \in \mathcal{S}) &= \operatorname{Pr}(\mathcal{E}^{c}) \operatorname{Pr}\left(\mathbf{y} \in \mathcal{S} | \mathcal{E}^{c} \right) + \operatorname{Pr}(\mathcal{E}) \operatorname{Pr}\left(\mathbf{y} \in \mathcal{S} | \mathcal{E}\right) \nonumber \\ 
    & \leq \delta' +   \operatorname{Pr}(\mathcal{E}) \operatorname{Pr}\left(\mathbf{y} \in \mathcal{S} | \mathcal{E}\right). \label{eq:centralProofStep}
\end{align}

To further upper bound Eqn. \eqref{eq:centralProofStep}, we use the following lemma.
\begin{lemma}
\label{lemma:conditonalBoundedByDP}
Let $\bar{p}_k=\operatorname{Pr}( \mathcal{E}_{k} | |\mu - \bar{\mu} | < t)$ and $c_{k}$ be some constant that depends on the local privacy mechanism, specifically for the Gaussian mechanism, we have $c_{k} \triangleq  { \gamma w_{k}  C_{k}}  \sqrt{2 \log ({1.25}/{\delta}) }$, where $C_{k}$ is the clipping parameter. The following inequality is true when the local mechanism satisfies $\left(\frac{c_{k}}{\sqrt{\bar{\mu} - t}}, \delta \right)$- feature DP:
\begin{align}
    &\operatorname{Pr}(\mathbf{y} \in \mathcal{S} | |\mu - \bar{\mu} | < t) \leq \bar{p}_k \delta_{k} \nonumber\\  
    &\hspace{0.3in} + \left[ 1 + \bar{p}_k \left( e^{\frac{c_{k}}{\sqrt{\bar{\mu} - t}}} -1 \right)  \right]  \operatorname{Pr}(\mathbf{y}_{-k} \in \mathcal{S} | |\mu - \bar{\mu} | < t)\nonumber,
\end{align}
where $\mu \triangleq \sum_{i = 1}^{K} \tau_{i} \sigma_{i}^{2}$, $\tau_{i} \sim \operatorname{Bern}(p_{i})$, $\bar{\mu} \triangleq \sum_{i=1}^{K} p_{i} \sigma_{i}^{2}$.
\end{lemma}
\begin{proof}
    The proof of the above lemma follows along similar lines as \cite{mohamed2021privacy} and is omitted for brevity.
\end{proof}

Using Lemma \ref{lemma:conditonalBoundedByDP}, we can bound \eqref{eq:centralProofStep} as follows:
\begin{align}
     &\operatorname{Pr}(\mathbf{y} \in \mathcal{S}) \nonumber\\
     & \leq \delta' + {\operatorname{Pr}(|\mu - \bar{\mu} | < t) } \Big[\bar{p}_k \delta \nonumber\\
     &+ \left[ \bar{p}_k \left( e^{\frac{c_{k}}{\sqrt{\bar{\mu} - t}}} -1 \right) + 1 \right]  \operatorname{Pr}(\mathbf{y}_{-k} \in \mathcal{S} | |\mu - \bar{\mu} | < t)  \Big] \nonumber \\ 
    & \overset{(a)} \leq \delta' + \bar{p}_k \delta + {\operatorname{Pr}(|\mu -\bar{\mu} | < t) }\nonumber\\
    &\hspace{8.5mm} \times \left[ \bar{p}_k \left( e^{\frac{c_{k}}{\sqrt{\bar{\mu} - t}}} -1 \right) + 1 \right] \frac{\operatorname{Pr}(\mathbf{y}_{-k} \in \mathcal{S} )}{\operatorname{Pr}(|\mu -\bar{\mu} | < t) }   \nonumber \\
    & \overset{(b)} \leq \delta' + \frac{p_k}{1-\delta'} \delta \nonumber\\
    &\hspace{8.3mm}+  \left[ \frac{p_k}{1-\delta'} \left( e^{\frac{c_{k}}{\sqrt{\bar{\mu} - t}}} -1 \right) + 1 \right] \operatorname{Pr}(\mathbf{y}_{-k} \in \mathcal{S})
\end{align}
where $(a)$ follows from total probability theorem and the fact that ${\operatorname{Pr}(|\mu - \bar{\mu} | < t) } \bar{p}_k \delta \leq \bar{p}_k\delta $; and $(b)$ follows from inequality $\bar{p}_k\leq p_k/(1-\delta')$ as mentioned before. This completes the proof of Theorem \ref{thm:feature_privacy_guarantee}.

\section{Proof of Theorem \ref{thm:lower_bound_classification}}

The aggregated signal received at the inference server can be re-written as follows:
\begin{align}
    \hat{\mathbf{f}} & = \mathbf{D} \sum_{k =1 }^{K} \tau_{k} w_{k}  \mathbf{W_{k}} \mathbf{f}_{k} + \mathbf{D} \sum_{k = 1 }^{K} \tau_{k} \mathbf{n}_{k} + \frac{1}{\gamma} \mathbf{D} \mathbf{m},
\end{align}
where $\tau_{k} \sim \operatorname{Bern}(p_{k})$.

We now start the proof of this theorem by analyzing and upper bounding the MSE as follows:

\begin{align}
 & \operatorname{MSE}    \triangleq \mathds{E} \left[ \|\hat{\mathbf{f}} - \mathbf{f}^{*}\|_{2}^{2} \right] \nonumber \\ 
 & = \mathds{E} \left[ \| \mathbf{D} \bigg( \sum_{k =1 }^{K} \tau_{k} w_{k}  \mathbf{W_{k}} \mathbf{f}_{k} + \sum_{k = 1 }^{K} \tau_{k} \mathbf{n}_{k} + \frac{1}{\gamma} \mathbf{m} \bigg) - \frac{1}{K} \sum_{k = 1}^{K} \mathbf{f}_{k}  \|_{2}^{2} \right] \nonumber \\ 
    & = \underbrace{\mathds{E} \left[ \| \sum_{k =1 }^{K} \bigg(\tau_{k} w_{k} \mathbf{D} \mathbf{W_{k}}  - \frac{1}{K} \mathbf{I} \bigg) \mathbf{f}_{k} \|_{2}^{2}  \right]}_{\mathcal{T}_{1}} \nonumber \\ 
    & \hspace{1.8in} + \underbrace{\mathds{E} \left[ \|\mathbf{D} \sum_{k = 1 }^{K} \tau_{k} \mathbf{n}_{k} + \frac{1}{\gamma} \mathbf{D} \mathbf{m}  \|_{2}^{2} \right]}_{\mathcal{T}_{2}},
\end{align}
where we can readily show that the cross terms are zeros due to independence and the fact that the noises are i.i.d. and drawn from Gaussian distribution with zero-mean. The second term can be upper bounded as 
\begin{align}
    \mathcal{T}_{2} \leq d \cdot \|\mathbf{D}\|_{F}^{2} \cdot  \left[ \sum_{k=1}^{K} p_{k} \sigma_{k}^{2} +  \frac{\sigma_{m}^{2}}{\gamma^{2}} \right]. \nonumber
\end{align}
Now, we focus on the first term $\mathcal{T}_{1}$:
\begin{align}
    \mathcal{T}_{1} & = \mathds{E} \left[ \| \sum_{k =1 }^{K} \bigg(\tau_{k} w_{k} \mathbf{D} \mathbf{W_{k}}  - \frac{1}{K} \mathbf{I} \bigg) \mathbf{f}_{k} \|_{2}^{2}  \right],
\end{align}
where $\mathbf{I}$ is the identity matrix. We next expand the norm inside the expaction as follows: 
\begin{align}
    & \mathcal{T}_{1}  = \sum_{k=1}^{K} \underbrace{\mathds{E} \left[ \|\bigg(\tau_{k} w_{k} \mathbf{D} \mathbf{W_{k}}  - \frac{1}{K} \mathbf{I} \bigg) \mathbf{f}_{k} \|_{2}^{2} \right]}_{\mathcal{T}_{1}'}  \nonumber \\ 
    & +  \underbrace{2 \sum_{k < j} \mathds{E} \left[ \mathbf{f}_{k}^{T} \bigg(\tau_{k} w_{k} \mathbf{D} \mathbf{W_{k}}  - \frac{1}{K} \mathbf{I} \bigg)^{T}  \bigg(\tau_{j} w_{j} \mathbf{D} \mathbf{W_{j}}  - \frac{1}{K} \mathbf{I} \bigg) \mathbf{f}_{j}^{T} \right]}_{\mathcal{T}_{1}''}.
\end{align}
We next simplify the diagonal terms as follows: 
\begin{align}
     & \mathcal{T}_{1}' = \mathbf{E} \left[ \|\tau_{k} w_{k} \mathbf{D} \mathbf{W}_{k} \mathbf{f}_{k} \|_{2}^{2} \right] \nonumber \\ 
     & \hspace{0.6in} - \frac{2}{K} \mathds{E} \left[ \left(\tau_{k} w_{k} \mathbf{D} \mathbf{W}_{k} \mathbf{f}_{k} \right)^{T} \mathbf{f}_{k} \right] + \frac{\|\mathbf{f}_{k}\|_{2}^{2}}{K^{2}}.
\end{align}
Taking the expectation:
\begin{align}
     & \mathcal{T}_{1}'   = p_{k} w_{k} \|\mathbf{D}\|_{F}^{2} \|\mathbf{W}_{k}\|_{F}^{2} \|\mathbf{f}_{k}\|_{2}^{2} \nonumber \\ 
     &  \hspace{0.6in}  - \frac{2 p_{k} w_{k}}{K} \|\mathbf{D} \|_{F}^{2} \|\mathbf{W}_{k}\|_{F}^{2} \|\mathbf{f}_{k}\|_{2}^{2} + \frac{\|\mathbf{f}_{k}\|_{2}^{2}}{K^{2}}. 
\end{align}
Simplifying: 
\begin{align}
    & \mathcal{T}_{1}'  = (w_{k}^{2} p_{k} - 2 w_{k} p_{k} + 1) \|\mathbf{D}\|_{F}^{2} \|\mathbf{W}_{k}\|_{F}^{2} \|\mathbf{f}_{k}\|_{2}^{2}.
\end{align}
We next simplify the cross terms as follows: 
\begin{align}
    & \mathcal{T}_{1}''  = 2 \sum_{k < j} \bigg[ p_{k} p_{j} w_{k} w_{j} \mathbf{f}_{k}^{T} \mathbf{W}_{k}^{T} \mathbf{D}^{T} \mathbf{D} \mathbf{W}_{j} \mathbf{f}_{j}  \nonumber \\ 
    & - \frac{p_{k} w_{k}}{K} \mathbf{f}_{k}^{T} \mathbf{W}_{k}^{T} \mathbf{D}^{T} \mathbf{D} \mathbf{f}_{j}  - \frac{p_{j} w_{j}}{K} \mathbf{f}_{k}^{T} \mathbf{W}_{j}^{T} \mathbf{D}^{T} \mathbf{D} \mathbf{f}_{j}  + \frac{1}{K^{2}} \mathbf{f}_{k}^{T} \mathbf{f}_{j}    \bigg].
\end{align}
Combining everything together, we get the upper bound on MSE, i.e., 
\begin{align}
& \operatorname{MSE}   \triangleq \mathds{E} \left[ \|\hat{\mathbf{f}} - \mathbf{f}^{*}\|_{2}^{2} \right] \nonumber \\ 
& \leq d \cdot \|\mathbf{D}\|_{F}^{2} \cdot  \left[ \sum_{k=1}^{K} p_{k} \sigma_{k}^{2}  +  \frac{\sigma_{m}^{2}}{\gamma^{2}} \right] \nonumber \\
& +  \sum_{k = 1}^{K} \bigg[ (w_{k}^{2} p_{k} - 2 w_{k} p_{k} + 1) \|\mathbf{D}\|_{F}^{2} \|\mathbf{W}_{k}\|_{F}^{2} \|\mathbf{f}_{k}\|_{2}^{2} \bigg] \nonumber \\
&   +  \sum_{k < j} \bigg[ (p_{k} p_{j} w_{k} w_{j} - p_{k} w_{k} - p_{j} w_{j} + 1) \mathbf{f}_{k}^{T} \mathbf{W}_{k}^{T} \mathbf{D}^{T} \mathbf{D} \mathbf{W}_{j} \mathbf{f}_{j} \bigg], \nonumber 
\end{align}

Now, we relate the MSE with the classification accuracy. It is worth while mentioning that the perturbed feature vector $\hat{\mathbf{f}}$ is  classified correctly if $\|\hat{\mathbf{f}} - \mathbf{f}^{*}\|_{2} \leq \Delta$, where $\Delta$ is the intrinsic classification margin of the pre-trained classifier. Next, we can easily show that the classification accuracy is lower bounded as 
\begin{align}
    P(\hat{l} = l^{*} )  & \geq P_{0} \cdot \operatorname{Pr}(\|\hat{\mathbf{f}} - \mathbf{f}^{*}\|_{2}  < \Delta) \nonumber \\
    & =   P_{0} \cdot \operatorname{Pr}(\|\hat{\mathbf{f}} - \mathbf{f}^{*}\|_{2}^{2}  < \Delta^{2}) \nonumber \\ 
    & \geq     P_{0} \cdot \left( 1 - \frac{\operatorname{MSE}}{\Delta^{2}} \right).
\end{align}
This completes the proof of Theorem \ref{thm:lower_bound_classification}.

\section{Optimizing the Features Weights}

\begin{lemma}\label{thm:optimal_weights} The optimal weight $w_{k} \geq 0$ for the $k$th device that satisfies the privacy constraint in \eqref{eqn:feature_privacy_guarantee} is picked as
\begin{align}
    w_k^{*} & = \min \bigg\{ w_{k,\max}, \frac{T_{1,*}+ T_{2,*} + \nu}{2 p_k \|\mathbf{D}\|_{F}^{2} \|\mathbf{W}_k\|_{F}^{2} \|\mathbf{f}_k\|_{2}^{2}}\bigg\}, \label{eqn:stationarity_condition_eqn}
\end{align}
where,
\begin{align}
    T_{1,*} & = - \sum_{k < j} p_k (p_{j} w_{j}^{*} - 1) \mathbf{f}_k^T \mathbf{W}_k^T \mathbf{D}^T \mathbf{D} \mathbf{W}_j \mathbf{f}_j, \nonumber \\ 
T_{2, *} & = 2 p_k \|\mathbf{D}\|_{F}^{2} \|\mathbf{W}_k\|_{F}^{2} \|\mathbf{f}_k\|_{2}^{2}.
\end{align}
The $ w_{k,\max}$ is the maximum permissible value of $w_{k}$ that satisfy Eqn. \eqref{eqn:feature_privacy_guarantee} and can be found numerically. The parameter $\nu$ is set such that $\sum_{k=1}^{K} w_{k} =1 $, and determined via the bisection method. 
\end{lemma}

We formulate a privacy-constrained optimization problem to minimize the MSE as follows:
\begin{align}
    \label{eq:weight_given_noise_opt_prob_uncorrelated}
    &\min_{\{w_{k}\}_{k=1}^{K}} \hspace{0.1cm} \operatorname{MSE} \hspace{0.1cm} \text{s.t.: }\hspace{0.1cm} \sum_{k=1}^{K} w_{k} =1, w_{k} \geq 0,\hspace{0.1cm} \forall \; k\in[K],\nonumber\\
    & \hspace{0.7cm} \text{Eqn.} \eqref{eqn:feature_privacy_guarantee} \hspace{0.1cm} \text{is satisfied} \Rightarrow w_{k} \leq w_{k, \max} \hspace{0.1cm}, \forall \; k\in[K].
\end{align}
 We next solve the optimization problem via the Lagrangian multiplier method. The Lagrangian function is given as follows:
\begin{align}
    \mathcal{L} & = d \cdot \|\mathbf{D}\|_{F}^{2} \cdot  \left[ \sum_{k=1}^{K} p_{k} \sigma_{k}^{2}  +  \frac{\sigma_{m}^{2}}{\gamma^{2}} \right] \nonumber \\
& +  \sum_{k = 1}^{K} \bigg[ (w_{k}^{2} p_{k} - 2 w_{k} p_{k} + 1) \|\mathbf{D}\|_{F}^{2} \|\mathbf{W}_{k}\|_{F}^{2} \|\mathbf{f}_{k}\|_{2}^{2} \bigg] \nonumber \\
&   +  \sum_{k < j} \bigg[ (p_{k} p_{j} w_{k} w_{j} - p_{k} w_{k} - p_{j} w_{j} + 1) \mathbf{f}_{k}^{T} \mathbf{W}_{k}^{T} \mathbf{D}^{T} \mathbf{D} \mathbf{W}_{j} \mathbf{f}_{j} \bigg] \nonumber \\
& - \sum_{k=1}^{K} \lambda_{k} (w_{k,\max} - w_{k}) - \sum_{k=1}^{K} \mu_{k} w_{k} - \nu \left( \sum_{k=1}^{K} w_{k} - 1\right). \label{eqn:lagrangian_function}
\end{align}
Evaluating the gradients, we get the following:
\begin{align}
    \frac{\partial \mathcal{L}}{\partial w_{k}} & =  (2 w_{k} p_{k} - 2 p_{k}) \|\mathbf{D}\|_{F}^{2} \|\mathbf{W}_{k}\|_{F}^{2} \|\mathbf{f}_{k}\|_{2}^{2} \nonumber \\ 
    &  + \sum_{k < j} \bigg[ (p_{k} p_{j}  w_{j} - p_{k}) \mathbf{f}_{k}^{T} \mathbf{W}_{k}^{T} \mathbf{D}^{T} \mathbf{D} \mathbf{W}_{j} \mathbf{f}_{j} \bigg] \nonumber  \\
    & + \lambda_{k}  - \mu_{k} - \nu, \\ 
     \frac{\partial \mathcal{L}}{\partial \lambda_{k}} & = w_{k,\max} - w_{k}, \\
     \frac{\partial \mathcal{L}}{\partial \mu_{k}} &  = - w_{k},      \frac{\partial \mathcal{L}}{\partial \nu} = 1 - \sum_{k=1}^{K} w_{k}. 
\end{align}

Applying the KKT conditions, from the stationarity condition we get the following expression for $w_{k}$:
\begin{align}
    \bar{w}_k & = \max \bigg\{ 0, \frac{T_{1,*}+ T_{2,*} + T_{3,*}}{2 p_k \|\mathbf{D}\|_{F}^{2} \|\mathbf{W}_k\|_{F}^{2} \|\mathbf{f}_k\|_{2}^{2}}\bigg\}, \label{eqn:stationarity_condition_eqn}
\end{align}
where,
\begin{align}
    T_{1,*} & = - \sum_{j \neq k} p_k (p_{j} w_{j} - 1) \mathbf{f}_k^T \mathbf{W}_k^T \mathbf{D}^T \mathbf{D} \mathbf{W}_j \mathbf{f}_j, \nonumber \\ 
T_{2, *} & = 2 p_k \|\mathbf{D}\|_{F}^{2} \|\mathbf{W}_k\|_{F}^{2} \|\mathbf{f}_k\|_{2}^{2}, \nonumber \\ 
T_{3,*} & = - \lambda_{k} + \mu_k + \nu.
\end{align}
The following condition holds
\begin{align}
    w_k & = \min \bigg\{ w_{k,\max},  \left[\frac{T_{1,*}+ T_{2,*} + \nu}{2 p_k \|\mathbf{D}\|_{F}^{2} \|\mathbf{W}_k\|_{F}^{2} \|\mathbf{f}_k\|_{2}^{2}} \right]^{+}\bigg\}, 
\end{align}
where $[x]^{+} = \max\{0, x\} $, $ w_{k,\max}$ is the maximum permissible value of $w_{k}$ that satisfy Eqn. \eqref{eqn:feature_privacy_guarantee} and can be found numerically. The parameter $\nu$ is set such that $\sum_{k=1}^{K} w_{k} =1 $, and determined via the bisection method.

As shown in the above, Eqn. \eqref{eqn:stationarity_condition_eqn} involves terms that contain both  $w_{k}$  and the sum over $j \neq k$. We propose a Gauss–Seidel method \cite{boyd2004convex} to iteratively optimize for \(w_{k}\) while satisfying the KKT conditions. This leads us to the following update rule:

\begin{figure}[t]
	\centering
    {\includegraphics[width=0.75\columnwidth]{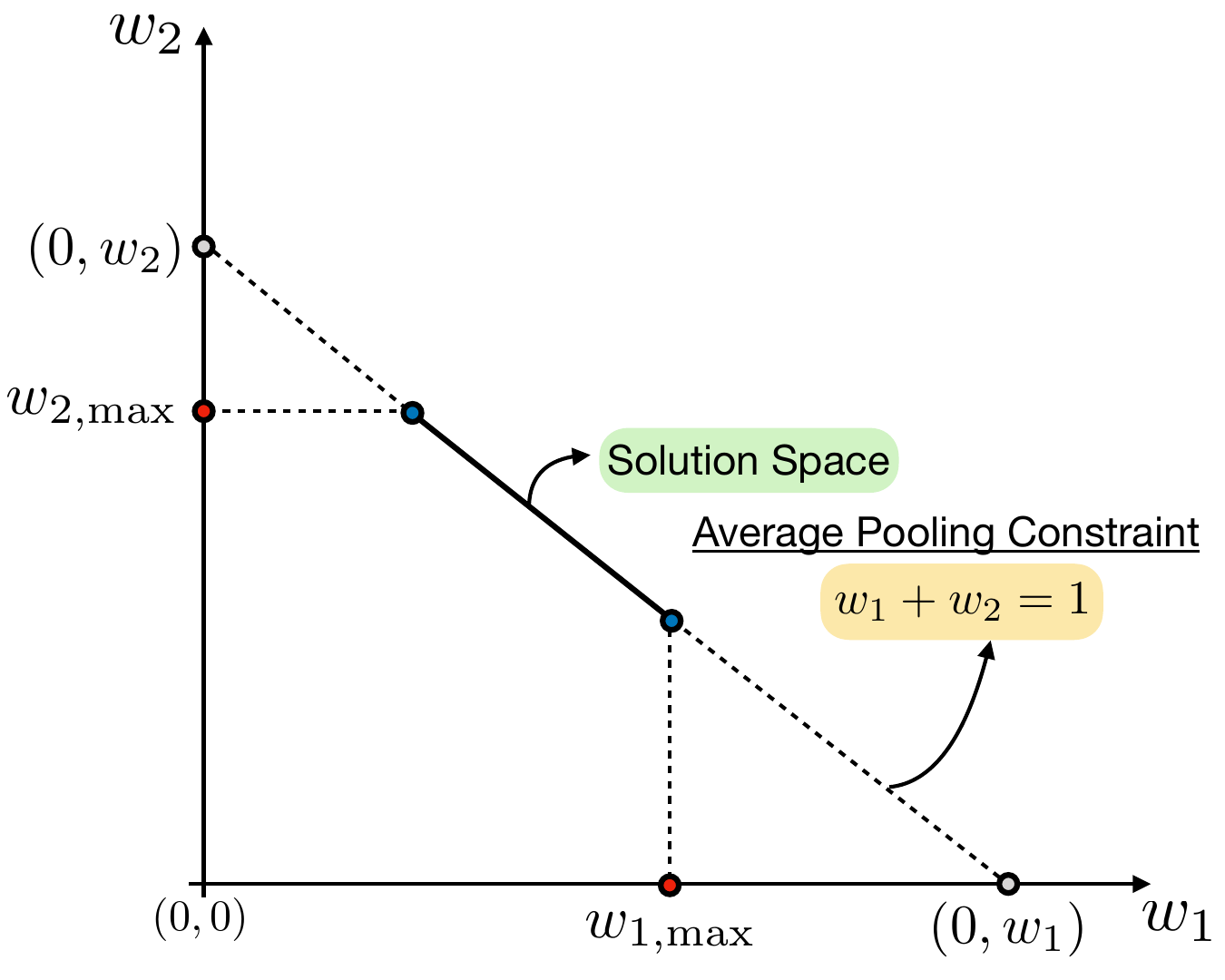}}
    \caption{\small{Feasible solution for the weight optimization for $K = 2$ devices. The maximum weights $w_{k}$ are determined by solving \eqref{eqn:feature_privacy_guarantee} numerically given $\epsilon_{k}$. }}
    \label{fig:feasible_solution}
    \vspace{-10pt}
\end{figure}

\begin{align}
    w_k^{(t)} & = \min \bigg\{ w_{k, \max}, \frac{T_{1,*}^{(t-1)} + T_{2,*} + \nu}{2 p_k \|\mathbf{D}\|_{F}^{2} \|\mathbf{W}_k\|_{F}^{2} \|\mathbf{f}_k\|_{2}^{2}}\bigg\},
\end{align}
 where
 \begin{align}
   T_{1,*}^{(t-1)} = - \sum_{j \neq k} p_k (p_{j} w_{j}^{(t-1)} - 1) \mathbf{f}_k^T \mathbf{W}_k^T \mathbf{D}^T \mathbf{D} \mathbf{W}_j \mathbf{f}_j.    
 \end{align}

\section{Proof of Theorem \ref{thm:sorting_correctness}}

For notational simplicity, we omit the superscript from $\sigma_{k}^{(0)}$ and use $\sigma_{k}$ throughout the subsequent analysis.

To describe the device selection process using the Law of Total Probability, we consider two scenarios: $(1)$ The true set \(\mathcal{K}\) of devices with the lowest \( u_k \) values is correctly identified (i.e., the event \(\mathcal{A}\) occurs), and $(2)$ The true set \(\mathcal{K}\) is not correctly identified (i.e., the event \(\mathcal{A}^c\) occurs, where \(\mathcal{A}^c\) is the complement of \(\mathcal{A}\)). Using the Law of Total Probability, the probability that device \( k \) is selected can be written as:

\begin{align}
    \operatorname{Pr}(k \in \tilde{\mathcal{K}}) = \operatorname{Pr}(k \in \tilde{\mathcal{K}} | \mathcal{A}) \operatorname{Pr}(\mathcal{A}) + \operatorname{Pr}(k \in \tilde{\mathcal{K}} | \mathcal{A}^c) \operatorname{Pr}(\mathcal{A}^c),
\end{align}
where  \(\tilde{\mathcal{K}}\) is the set of devices that are selected based on the noisy scores, \(\operatorname{Pr}(k \in \tilde{\mathcal{K}} | \mathcal{A})\) is the probability that device \( k \) is selected given that the true set \(\mathcal{K}\) is identified correctly, \(\operatorname{Pr}(\mathcal{A})\) is the probability that the true set \(\mathcal{K}\) is correctly identified, \(\operatorname{Pr}(k \in \tilde{\mathcal{K}} | \mathcal{A}^c)\) is the probability that device \( k \) is selected given that the true set \(\mathcal{K}\) is not identified correctly, and \(\operatorname{Pr}(\mathcal{A}^c)\) is the probability that the true set \(\mathcal{K}\) is not correctly identified.

The term \(\operatorname{Pr}(k \in \tilde{\mathcal{K}} | \mathcal{A}) \operatorname{Pr}(\mathcal{A})\) describes the contribution to the probability of selecting device \( k \) when the true set \(\mathcal{K}\) is correctly identified. If \(\mathcal{A}\) occurs, then the sorting process has accurately ranked the devices according to their true scores \( u_k \). This term is generally higher for devices with low \( u_k \). The term \(\operatorname{Pr}(k \in \tilde{\mathcal{K}} | \mathcal{A}^c) \operatorname{Pr}(\mathcal{A}^c)\) captures the probability that device \( k \) is selected when the true set \(\mathcal{K}\) is not correctly identified due to the influence of noise. This term is generally lower for well-ranked devices when \(\mathcal{A}^c\) occurs, as the noise might disrupt the ranking order.

To provide a more precise characterization of the probabilities involved in the selection of device \( k \), we examine the terms \(\operatorname{Pr}(k \in \tilde{\mathcal{K}} | \mathcal{A})\), \(\operatorname{Pr}(\mathcal{A})\), and \(\operatorname{Pr}(k \in \tilde{\mathcal{K}} | \mathcal{A}^c)\).  When \(\mathcal{A}\) occurs, if device \( k \in \mathcal{K}\), then \(\operatorname{Pr}(k \in \tilde{\mathcal{K}} | \mathcal{A}) = \operatorname{Pr}(k \in {\mathcal{K}} | \mathcal{A})  = 1\), otherwise it is 0. When \(\mathcal{A}^c\) occurs, the noise affects the selection.  The probability that a specific device \( k \) is included in the set \(\tilde{\mathcal{K}}\) of devices with the lowest noisy scores is given by:
\begin{align}
    \Pr(k \in \tilde{\mathcal{K}} | \mathcal{A}^c) &= \Phi\left(\frac{u_k - \min_{j \in \mathcal{K}^c} u_j}{\sqrt{\sigma_k^2 + \sigma_j^2}}\right),
\end{align}
where \(\mathcal{K}^{c} = [K] \backslash \mathcal{K}\) denotes the complement of \(\mathcal{K}\). Thus, the probability \( p_{k} \) that device \( k \) is selected can be written as:
\begin{align}
    p_{k} &= \Pr(k \in \tilde{\mathcal{K}}) \nonumber \\
          &= \Pr(\mathcal{A}) \cdot \mathbb{I}(k \in \mathcal{K}) + \Pr(k \in \tilde{\mathcal{K}} | \mathcal{A}^c) \cdot \Pr(\mathcal{A}^{c}).
\end{align}

Furthermore, a lower bound on \( p_{k} \) is given by:
\begin{align}
    p_{k} \geq \Pr(\mathcal{A}) \cdot \mathbb{I}(k \in \mathcal{K}).
\end{align}

Similarly, we obtain an upper bound for \( p_{k} \) as follows:
\begin{align}
    p_{k} \leq \Pr(\mathcal{A}) \cdot \mathbb{I}(k \in \mathcal{K}) + \Pr(k \in \tilde{\mathcal{K}} | \mathcal{A}^c), \label{eqn:law_total_prob}
\end{align}
where \(\mathbb{I}(k \in \mathcal{K})\) is an indicator function that is 1 if device \(k\) is in the true set \(\mathcal{K}\), and 0 otherwise, and \(\Pr(\mathcal{A}^c) = 1 - \Pr(\mathcal{A})\) is the probability that the ordering of noisy scores differs from the ordering of true scores.

 We next derive an upper bound for the probability \(\Pr(\mathcal{A})\), where \(\mathcal{A}\) represents the event that the true set \(\mathcal{K}\), consisting of the devices with the lowest true scores \(u_k\), is correctly identified based on their noisy scores \(\tilde{u}_k = u_k + n_k\). Here, \(n_k\) represents the noise with variance \(\sigma_k^2\). Additionally, we derive and analyze the probability \(\Pr(k \in \tilde{\mathcal{K}} | \mathcal{A}^c)\), which represents the likelihood of a specific device being selected when the true set is not correctly identified.

To derive the upper bound, we decompose \(\Pr(\mathcal{A})\) into the probability that the ordering of noisy scores matches the ordering of true scores:

\[
\mathcal{A} = \left\{ \tilde{u}_{(1)} \leq \tilde{u}_{(2)} \leq \ldots \leq \tilde{u}_{(K)} \right\}.
\]

Assuming independence between the orderings of adjacent noisy scores, we can bound \(\Pr(\mathcal{A})\) as follows:

\begin{align}
  \Pr(\mathcal{A}) \leq \prod_{k=1}^{K-1} \Pr(\tilde{u}_{(k)} \leq \tilde{u}_{(k+1)}),  
\end{align}
where each term \(\Pr(\tilde{u}_{(k)} \leq \tilde{u}_{(k+1)})\) represents the probability that noise does not disrupt the order between adjacent devices \(k\) and \(k+1\). Each pairwise term \(\Pr\left(\tilde{u}_{(k)} \leq \tilde{u}_{(k+1)}\right)\) can be further expressed as:
\begin{align}
    \Pr\left(\tilde{u}_{(k)} \leq \tilde{u}_{(k+1)}\right) = \Pr\left(n_{(k+1)} - n_{(k)} \leq u_{(k)} - u_{(k+1)}\right),
\end{align}
where \(n_{(k+1)} - n_{(k)} \sim \mathcal{N}(0, \sigma_{(k+1)}^2 + \sigma_{(k)}^2)\) is normally distributed with mean 0 and variance \(\sigma_{(k+1)}^2 + \sigma_{(k)}^2\). To simplify the analysis, we define \(\Psi\) as the minimum gap between adjacent true scores, i.e.,  $\Psi \triangleq \min_{1 \leq k \leq K-1} (u_{(k+1)} - u_{(k)})$.

Using this definition, we can bound each pairwise probability as follows:
\begin{align}
\Pr(\tilde{u}_{(k)} \leq \tilde{u}_{(k+1)}) \leq \Phi\left(\frac{\Psi}{\sqrt{\sigma_{(k+1)}^2 + \sigma_{(k)}^2}}\right),
\end{align}
where \(\Phi(\cdot)\) denotes the cumulative distribution function (CDF) of the standard normal distribution. Substituting the bounds into the product, we obtain the overall upper bound on \(\Pr(\mathcal{A})\):
\begin{align}
\Pr(\mathcal{A}) \leq \prod_{k=1}^{K-1} \Phi\left(\frac{\Psi}{\sqrt{\sigma_{(k+1)}^2 + \sigma_{(k)}^2}}\right).
\end{align}
The above expression provides an upper bound on the probability that the ordering of the noisy scores matches the ordering of the true scores, indicating the likelihood of correctly identifying the set \(\mathcal{K}\).

When the true set \(\mathcal{K}\) is not correctly identified (i.e., \(\mathcal{A}^c\) occurs), the probability that device \(k\) is selected depends on how its noisy score compares to those of other devices. Specifically, we consider the difference between the true score \(u_k\) of device \(k\) and the smallest true score among devices in \(\mathcal{K}^c\):

\begin{align}
    \Pr(k \in \tilde{\mathcal{K}} | \mathcal{A}^c) = \Pr\left(\tilde{u}_k \leq \min_{j \in \mathcal{K}^c} \tilde{u}_j \mid \mathcal{A}^c\right).
\end{align}

Since \(\tilde{u}_k = u_k + n_k\) and \(\tilde{u}_j = u_j + n_j\), this event can be written as:

\begin{align}
\Pr(k \in \tilde{\mathcal{K}} | \mathcal{A}^c) = \Pr\left(n_k - \min_{j \in \mathcal{K}^c} n_j \leq \min_{j \in \mathcal{K}^c} u_j - u_k\right).
\end{align}

Assuming \(n_k\) and \(n_j\) are independent and normally distributed, we get:

\begin{align}
\Pr(k \in \tilde{\mathcal{K}} | \mathcal{A}^c) = \Phi\left(\frac{u_k - \min_{j \in \mathcal{K}^c} u_j}{\sqrt{\sigma_k^2 + \sigma_j^2}}\right),
\end{align}
where \(\Phi\) is the standard normal CDF, and \(\sigma_k^2\) and \(\sigma_j^2\) are the variances of the noise terms for devices \(k\) and \(j\). Substituting the expressions for \(\Pr(\mathcal{A})\) and \(\Pr(k \in \tilde{\mathcal{K}} | \mathcal{A}^c)\) into Eqn. \eqref{eqn:law_total_prob}, we obtain:
\begin{align}
p_{k} & \leq \prod_{k=1}^{K-1} \Phi\left(\frac{\Psi}{\sqrt{\sigma_{(k+1)}^2 + \sigma_{(k)}^2}}\right) \cdot \mathbb{I}(k \in \mathcal{K}) \nonumber \\
& \hspace{0.3in} + \Phi\left(\frac{u_k - \min_{j \in \mathcal{K}^c} u_j}{\sqrt{\sigma_k^2 + \sigma_j^2}}\right).
\end{align}
This completes the proof of the theorem.

\subsection{Auxiliary Lemma}

\begin{lemma} The proposed sorting scheme is $(\epsilon_{k}, \delta)$-feature DP and returns the true sorting with probability at least 
\begin{align}
  \operatorname{Pr}(\text{true sorting}) \geq  \max \left\{ 0, 1 -  \sum_{k=1}^{K-1}  \Phi \left(\frac{-\Psi}{\sqrt{\sigma_{k}^{2} + \sigma_{k+1}^{2}}} \right) \right\},
\end{align}
where $\Phi(\cdot)$ is the cumulative CDF of Gaussian distribution, and $\Psi \triangleq \min_{1 \leq k \leq K-1} u_{(k+1)} - u_{(k)}$.
\end{lemma}

\begin{proof}
    Denote $\mathcal{A}$ the event the sorting server returns true set $\mathcal{K} \subseteq [K]$ of edge devices with the lowest $u_{k}$'s.  The noisy sorted confidence score is denoted by $\tilde{u}_{(k)}$, where $\tilde{u}_{(1)} \leq \tilde{u}_{(2)} \leq \cdots \leq \tilde{u}_{(K)}$. We next lower bound the probability on the event $\mathcal{A}$ as follows:
\begin{align}
    \operatorname{Pr}(\mathcal{A})  & \geq \operatorname{Pr}(\tilde{u}_{(1)} \leq \tilde{u}_{(2)} \leq \cdots \leq \tilde{u}_{(K)}) \nonumber \\ 
    & \geq \operatorname{Pr}(n_{1} - n_{2} \leq \Psi, \cdots, n_{K-1} - n_{K} \leq \Psi), 
\end{align}
where $\Psi \triangleq \min_{1 \leq k \leq K-1} u_{(k+1)} - u_{(k)}$. Using Bonferroni lower bound \cite{edition2002probability}, we have 
\begin{align}
    \operatorname{Pr}(\mathcal{A}) & \geq 1 - \sum_{i = 1}^{K-1} \left[1 - \operatorname{Pr}(n_{k} - n_{k+1} \leq \Psi ) \right] \nonumber \\ 
    & = 1 - \sum_{k=1}^{K-1} \operatorname{Pr}(n_{k} - n_{k+1} \leq - \Psi) \nonumber \\ 
    & = 1 -  \sum_{k=1}^{K-1}  \Phi \left(\frac{-\Psi}{\sqrt{\sigma_{k}^{2} + \sigma_{k+1}^{2}}} \right),
\end{align}
where $\Phi(\cdot)$ is the cumulative CDF of Gaussian distribution. This completes the proof of the lemma.
\end{proof}

\end{document}